\documentclass[aps,amssymb,amsmath,notitlepage,superscriptaddress,pra,twocolumn]{revtex4-1}

\usepackage{setspace}
\usepackage{amsfonts}
\usepackage{graphicx,graphics,epsfig,times,bm,bbm,amssymb,amsmath,amsfonts,mathrsfs}
\usepackage[normalem]{ulem}
\usepackage{wrapfig}
\usepackage{setspace}
\usepackage{subfigure}
\usepackage{dsfont}
\usepackage{braket}
\usepackage[pdftex]{color}
\usepackage[pdfstartview=FitH]{hyperref}
\usepackage{upgreek }
\usepackage[x11names]{xcolor}
\usepackage{tikz}
\usepackage{natbib}
\usepackage{chngcntr}
\usepackage{stmaryrd}
\usepackage[makeroom]{cancel}

\newtheorem{theorem}{Theorem}

\newtheorem{corollary}[theorem]{Corollary}

\newenvironment{proof}[1][Proof]{\noindent\textbf{#1.} }{\ \rule{0.5em}{0.5em}}

\newtheorem{mylemma}{Lemma}

\newtheorem{mytheorem}{Theorem}

\newcommand{\bes} {\begin{subequations}}
\newcommand{\ees} {\end{subequations}}
\newcommand{\bea} {\begin{eqnarray}}
\newcommand{\eea} {\end{eqnarray}}
\newcommand{\beq}{\begin{equation}}
\newcommand{\eeq}{\end{equation}}

\newcommand{\mc}{\mathcal}

\def\>{\rangle}
\def\<{\langle}
\def\Tr{\textrm{Tr}}

\def\tot{\textrm{tot}}
\def\sys{\textrm{S}}
\def\inter{\textrm{I}}
\def\B{\textrm{R}}
\def\SB{\textrm{SR}}
\def\id{\textrm{id}}
\def\IS{\textrm{IS}}
\def\leak{\textrm{leak}}

\def\tru{\textrm{trunc}}
\def\dec{\textrm{dec}}
\renewcommand{\min}{\textrm{min}}
\renewcommand{\max}{\textrm{max}}

\newcommand{\ketbra}[2]{|{#1}\>\!\<#2|}
\newcommand{\bracket}[2]{\<{#1}|{#2}\>}

\newcommand{\ignore}[1]{}

\begin{document}

\title{Quantum speed limits for leakage and decoherence}

\begin{abstract}
We introduce state-independent, non-perturbative Hamiltonian quantum speed limits for population leakage and fidelity loss, for a gapped open system interacting with a reservoir.  These results hold in the presence of initial correlations between the system and the reservoir, under the sole assumption that their interaction and its commutator with the reservoir Hamiltonian are norm-bounded. The reservoir need not be thermal and can be time-dependent. We study the significance of energy mismatch between the system and the local degrees of freedom of the reservoir which directly interact with the system. We demonstrate that, in general, by increasing the system gap we may reduce this energy mismatch, and consequently drive the system and the reservoir into resonance, which can accelerate fidelity loss, irrespective of the thermal properties or state of the reservoir. This implies that quantum error suppression strategies based on increasing the gap are not uniformly beneficial.  Our speed limits also yield an elementary lower bound on the relaxation time of spin systems.
\end{abstract}

\begin{abstract}
We introduce state-independent, non-perturbative Hamiltonian quantum speed limits for population leakage and fidelity loss, for a gapped open system interacting with a reservoir.  These results hold in the presence of initial correlations between the system and the reservoir, under the sole assumption that their interaction and its commutator with the reservoir Hamiltonian are norm-bounded. The reservoir need not be thermal and can be time-dependent. We study the significance of energy mismatch between the system and the local degrees of freedom of the reservoir which directly interact with the system. We demonstrate that, in general, by increasing the system gap we may reduce this energy mismatch, and consequently drive the system and the reservoir into resonance, which can accelerate fidelity loss, irrespective of the thermal properties or state of the reservoir. This implies that quantum error suppression strategies based on increasing the gap are not uniformly beneficial.  Our speed limits also yield an elementary lower bound on the relaxation time of spin systems.
\end{abstract}

\author{Iman Marvian}
\affiliation{Center for Quantum Information Science and \& Technology}
\affiliation{Department of Physics}
\author{Daniel A. Lidar}
\affiliation{Center for Quantum Information Science and \& Technology}
\affiliation{Department of Physics}
\affiliation{Department of Electrical Engineering}
\affiliation{Department of Chemistry\\University of Southern California, Los Angeles, California 90089, USA}
\maketitle

Quantum speed limits (QSLs) answer the fundamental question of how fast a quantum system can evolve, and have numerous applications, e.g., in quantum computation, control, and metrology. Traditionally, they characterize the minimum amount of time required for a quantum state of a closed quantum system to evolve to an orthogonal state. Mandelstam \& Tamm (MT) \cite{QSL_MT} first showed that this time is lower bounded by the inverse of the standard deviation of the Hamiltonian. Margolus \& Levitin (ML) \cite{Margolus:98} gave a different bound involving the inverse of the mean of the Hamiltonian, and the two bounds were subsequently unified \cite{QSL_Toffoli}.  These results led to numerous  applications  and extensions which go beyond the traditional QSLs, and consider, e.g., the minimum time for optimal control, or for implementing a unitary gate in quantum computation    
\cite{QSL_Fleming,QSL_Bha83,QSL_Vaidman,Pfeifer:1993fk,QSL_Pfeifer_RMP,QSL_Lloyd, QSL_Lloyd2, QSL_Zych, QSL_Chau,Deffner_uncertainty, ashhab2012speed, murphy2010communication,caneva2009optimal,schulte2005optimal,muller2011optimizing,zeier2008time,khaneja2007shortest}.

In this Letter we are concerned with QSLs for open quantum systems \cite{Breuer:book}, a question that has attracted significant recent attention \cite{QSL_Davidov_13,QSL_Plenio_13,QSL_Deffner_13,QSL_14_SciRep}. While earlier work focused on generalizing the MT or ML-bounds to the open system setting, here we present state-independent, non-perturbative Hamiltonian QSL bounds for population leakage and fidelity loss, for a gapped open system interacting with a reservoir. The assumptions behind the results we present here are also different and independent from those behind previous such bounds. First, we make the (often natural) assumption that the system's initial state is restricted to an energy sector which is separated from the rest of the spectrum by a nonzero gap $\Delta E$, e.g., the ground subspace in various quantum information processing applications.
Second, our bounds are independent of the state of the system or reservoir, and in particular, remain valid in the presence of initial correlations between the system and the reservoir. Third, our bounds are obtained purely at the Hamiltonian level.  
Thus, unlike most other open system QSL bounds \cite{QSL_Davidov_13,QSL_Plenio_13,QSL_Deffner_13,QSL_14_SciRep}, we do not assume that the system's evolution is governed by a master equation or a completely positive channel.   

Our key result is given in Eqs.~\eqref{eq:tau_leak} and \eqref{eq:tau_fid} below, and comprises fundamental QSL bounds on decoherence and leakage times, expressed in terms of $\Delta E$ and the bounded norms of the interaction Hamiltonian and its commutator with the reservoir Hamiltonian, without assuming that the reservoir Hamiltonian is norm-bounded.   Note that  applying the traditional closed system QSL bounds to the system and reservoir together in general yields bounds which are rather loose and independent of the gap $\Delta E$ \cite{comment-multi}.

Given the very general assumptions behind our QSLs, they have a wide range of applicability similar to the previously known QSLs, including relaxation in many-body spin systems and limitations of control via a remote controller.  
The primary application on which we focus is quantum error suppression, where the goal is to slow down the loss of fidelity relative to some desired system state, e.g., in the context of quantum information processing tasks \cite{nielsen2000quantum,Lidar-Brun:book}.
A common strategy to achieve fidelity enhancements is to use or generate energy gaps (e.g., \cite{Kitaev:97,Dennis:02,Bacon:01,PhysRevA.74.052322,PAL:13,Bookatz:2014uq,Terhal:2015dq}).  
Therefore, after deriving our QSLs we study the dependence of the speed of decoherence and leakage on $\Delta E$. This enables us to find a general lower bound on the timescale for leakage. As expected, we find that in the $\Delta E \to \infty$ limit the probability of leakage at any finite time goes to zero, and moreover, that if the {error detection condition} \cite{Knill:1997kx} holds then in this limit the state retains its fidelity and remains unaffected by the environment. However, we demonstrate that such protection is not guaranteed when $\Delta E$ is finite. Namely, by analyzing a spin system model, we show that increasing the gap can in fact accelerate fidelity loss and decoherence, essentially because of a resonance between the system and the reservoir. This means that protection via increasing energy gaps can be counterproductive \cite{comment-anti-Zeno}.

\textit{Technical results}.---%
Consider a system S coupled to a reservoir R with the total Hamiltonian $H_\tot(t)={H}_\sys+ {H}_\B(t)+{H}_\inter $ where $[{H}_\sys, {H}_\B(t)]=0$ and the interaction satisfies $\|H_\inter \| < \infty$ (we use the operator norm $\|\cdot\|$, i.e., the largest singular value; we also use $\hbar=1$ units throughout).   An important class of examples are spin-bath models \cite{Prokofev:00}.  We denote the time-dependent joint system-reservoir state evolving under $H_\tot(t)$ by $\rho_\SB(t)$ and the reduced state of the system at time $t$ by $\rho(t) = \Tr_\B[\rho_\SB(t)]$. 
Let $\mc{C}$ be the subspace of the system Hilbert space spanned by  the eigenstates of $H_\sys$ whose energies lie in the interval $\mc{I}\subseteq \mathbb{R}$, which includes at least one eigenvalue of  $H_\sys$. 
Let $P_\mc{C}$ be the projector onto $\mc{C}$, and $Q_\mc{C}\equiv I-P_\mc{C}$ be the projector onto the orthogonal subspace $\mc{C}^\perp$. Thus $[ P_\mc{C}, H_\sys ]=0$. 
Let $\delta E$ denote the energy spread in $\mc{C}$, i.e., the difference between the minimum and maximum eigenvalues of $H_\sys$ in $\mc{I}$.
Let $\Delta E$ denote the gap between the energy levels of $H_\sys$ inside and outside $\mc{C}$ (i.e., if $\lambda_1$ and $\lambda_2$ are two distinct eigenvalues of $H_\sys$ such that $\lambda_1\in \mc{I}$ but $\lambda_2\notin \mc{I}$, then $|\lambda_1-\lambda_2|\ge\Delta E$). 
Initially we assume $\Delta E > 2\| H_\inter \|$, which guarantees that there is a separation between the system energies inside $\mc{I}$ and the rest of the spectrum, even in the presence of the interaction. 
This simplification is relevant because we are mostly interested in the large $\Delta E$ limit.  
Later, when we arrive at Eq.~\eqref{change-RF}, we present the general form of the result which relaxes this condition, and results in tighter bounds for small $t$, even when $\Delta E < 2\| H_\inter \|$.   Before we introduce our bounds, we define an important inverse timescale for open system dynamics, that will make repeated appearances:
\beq
\Omega(t) \equiv \frac{2 \|[H_\inter ,H_\B(t)]\|}{\Delta E-2 \|H_\inter \|}\ .
\eeq
We proceed to present and interpret our main results. All our results are given rigorous proofs in the Supplementary Material (SM) \cite{SM_CO}.
Unless stated otherwise, throughout we assume that the system state is initialized in $\mc{C}$, i.e., $\rho(0)=P_\mc{C}\rho(0)P_\mc{C}$. 

\textit{Leakage}.---%
Leakage is the process whereby the system state develops support in $\mc{C}^\perp$, which we quantify in terms of the leakage probability $p_\leak(t)\equiv\Tr\!\left[ \rho(t) Q_\mc{C}\right]$. 
Our first general result is an upper bound on $p_\leak(t)$, proved in the SM \cite{SM_CO}: 
\beq
p_\leak(t) \le    \left(\frac{{4}  \|H_\inter \|}{ \Delta E} + \int_0^t ds\  \Omega(s) \right)^{2} \stackrel{\!\!\!\!\!\!\Delta E \shortrightarrow \infty} {\xrightarrow{\hspace*{1cm}} 0} \ .
\label{bound_prob1}
\eeq
To explain this bound, note that the terms ${\|H_\inter \|}/{ \Delta E}$ and $ \int_0^t ds\  \Omega(s)$ correspond to two different sources of leakage: ${\|H_\inter \|}/{ \Delta E}$ determines how much $\mc{C}$ is rotated by the interaction $H_\inter $. The rotated eigenstates of the perturbed Hamiltonian can cause leakage relative to the eigenstates of the original Hamiltonian. Of course, this also happens in the closed systems, and this is why this term does not vanish for $H_\B(t)=0$, where the total system Hamiltonian becomes $H_\sys+H_\inter $. Since ${\|H_\inter \|}/{ \Delta E}$ is time-independent, it remains small and insignificant in the limit where the gap is large. 
The term $\int_0^t ds\  \Omega(s)$ is more interesting. In particular, in the case of   time-independent  $H_\B(t)=H_\B$, where the total energy of the system and reservoir is a conserved quantity, $\|[H_\B,H_\inter]\|$ can be interpreted as the maximum rate of change of energy of reservoir.
Then, in the special case where $\mc{C}$ is the bottom (top) energy sector, $\Omega^{-1}$ can be interpreted as the minimum time the reservoir needs to transfer (absorb) the required energy to move the system from $\mc{C}$ to  $\mc{C}^\perp$ (see the SM \cite{SM_CO}).

\textit{Fidelity}.---%
We compare the instantaneous ``actual" state $\rho(t)$ and the ``ideal" system state  $\rho_\id(t) = e^{-i t H_\sys} \rho(0) e^{i t H_\sys}$ 
using their Uhlmann {fidelity} \cite{Uhlmann, Fidelity_Jozsa} $F\!\left[\rho(t),\rho_\id(t)\right]\equiv \|\sqrt{\rho(t)}\sqrt{ \rho_\id(t)}\|_1$ ($\|\cdot\|_1$ is the trace norm) and their {Bures angle}  $\Theta(t)\equiv\arccos{\left(F[\rho(t),\rho_\id(t)]\right)}$, a generalization to mixed states of the angle between two pure states \cite{bures}. 
Let $P_0 \equiv P_\mc{C}\otimes I_\B$. We define the \emph{induced splitting} by $H_\inter $ on $\mc{C}$ as
\beq
\IS(P_0 H_\inter P_0)\equiv \min_{K_\B\in  \textrm{Herm}(\mc{H}_\B)} \|P_0 H_\inter P_0-P_\mc{C}\otimes K_\B \| \ ,
\label{IS:Def}
\eeq
where the minimization is over the Hermitian operators acting on the reservoir Hilbert space $\mc{H}_\B$. This quantity can be interpreted as the strength of the effective interaction between the code subspace and the reservoir in the lowest order of perturbation theory. 
It exists because the reservoir can couple to different states in the subspace $\mathcal{C}$ in different ways and this generally leads to decoherence, or, in special cases, to a modification of the system Hamiltonian, a potentially beneficial effect \cite{Zanardi:2014fr} (see the SM \cite{SM_CO}).
This term can be nonzero only when $\mc{C}$ is at least two-dimensional. 
We can now state our second general result, an infidelity upper-bound:
\begin{align}
\label{main_bound_open}
&\sin \frac{\Theta(t)}{2}= \frac{1}{\sqrt{2}}\sqrt{1-F\!\left[\rho(t),\rho_\id(t)\right]}\le \frac{{2} \|H_\inter \|}{ \Delta E}\\
& +  \int_0^t  ds\ \Omega(s)+ t\left[\frac{\IS(P_0 H_\inter P_0)}{2}  + \frac{{2} \|H_\inter \| (\delta E+\|H_\inter \|)}{ \Delta E} \right] \notag \ .
\end{align}
Related bounds have been obtained in Ref.~\cite{Bookatz:2014uq}. While bound~\eqref{main_bound_open} holds for $\Delta E> 2\|H_\inter\|$ and states initialized in $\mc{C}$, 
our third general result is a simple universal QSL bound which does not require either one of these assumptions:
\begin{align}
\label{general_bound}
\sin \frac{\Theta(t)}{2}\le \frac{t (\lambda_{\max}-\lambda_{\min})}{4} \le \frac{t \|H_\inter \|}{2}\ ,
\end{align}
where $\lambda_{\max}$ and $\lambda_{\min}$ are the maximum and minimum eigenvalues of  $H_\inter $, respectively.
This bound formalizes the standard intuition that the minimum relaxation time of an interacting system is determined by the inverse of the couplings. However, as we will show in an explicit example, our QSL bounds~\eqref{bound_prob1} and \eqref{main_bound_open} can lead to much stronger bounds on the relaxation time.

\textit{Quantum speed limits}.---%
The bounds we have presented directly lead to QSLs on open-system quantum evolution, as we show next. 
For simplicity, in the following we assume that $H_\B(t)=H_\B$. 
 
Let $\tau^\mc{C}_\leak$ denote the smallest time at which the probability of leakage from $\mathcal{C}$ exceeds a constant threshold $p_0\in (0,1)$. 
Then, it follows from bound~\eqref{bound_prob1} 
that in the large-gap limit (i.e., $\|H_\inter\|/\Delta E\ll p_0^{1/2}$) this timescale is lower-bounded by $\frac{\Delta E}{{2}\|[H_\inter ,H_\B]\|}  p_0^{1/2}$.
We can find a different lower bound on $\tau^\mc{C}_\leak$ using bound~\eqref{general_bound} together with the fact that $F\!\left[\rho(t),\rho_\id(t)\right]\le \sqrt{1-p_\leak(t)}$ (see the SM \cite{SM_CO}). Let $\tau_{\min}$ be the smallest time at which $F\!\left[\rho(t),\rho_\id(t)\right]$ drops below the threshold $(1-p_0)^{1/2}$ for an arbitrary initial state. This threshold convention guarantees $\tau^\mc{C}_\leak\ge \tau_{\min}$. Then   
bound~\eqref{general_bound} implies $\tau_{\min}\ge c(p_0)\|H_\inter\|^{-1}$, where $c(p_0) = [2\{1-(1-p_0)^{1/2}\}]^{1/2}$, and hence
\begin{align}
\label{eq:tau_leak}
\tau^\mc{C}_\leak &\ge \max\left\{c(p_0)\|H_\inter\|^{-1}\ ,\ p_0^{1/2} \frac{\Delta E}{{2}\|[H_\inter ,H_\B]\|}  \right\} \ .
\end{align}

Similarly, we can define $\tau^\mc{C}_\text{fid}$ to be the smallest time at which $F\!\left[\rho(t),\rho_\id(t)\right]$ drops below the threshold $(1-p_0)^{1/2}$.
For this choice of threshold we always have $\tau^\mc{C}_\text{fid}\le \tau^\mc{C}_\leak$. 
If we further assume the same large-gap limit and also that $\IS(P_0 H_\inter P_0)=0$ and $\delta E=0$, which is a relevant assumption in the context of error suppression, we find using bound~\eqref{main_bound_open} that
\beq
\tau^\mc{C}_\text{fid} \ge {c(p_0)}\ \max\left\{ \|H_\inter \|^{-1}\ ,\ \frac{\Delta E}{4(\|[H_\inter ,H_\B]\|+\|H_\inter \|^2)} \right\}   \ .
\label{eq:tau_fid}
\eeq  
Equations~\eqref{eq:tau_leak} and \eqref{eq:tau_fid} constitute our key new QSL bounds.
It follows from the definitions of the various timescales we have introduced, together with our result in the bound~\eqref{general_bound}, that
\beq
\label{all_bounds}
\tau^\mc{C}_\leak \ge \tau^\mc{C}_\text{fid} {\ge} \tau_{\min}\ge  c(p_0) \|H_\inter \|^{-1}\ .    
 \eeq
The above bounds on $\tau^\mc{C}_\leak$, $\tau^\mc{C}_\text{fid}$ and $ \tau_{\min} $ are all first-order in $\|H_\inter \|^{-1}$.  
On the other hand, any master equation derived under the Born-Markov approximation (BMA) is necessarily second-order in the coupling strength \cite{Breuer:book}.
Therefore, these QSL time scales, or more generally any open-system behavior which occurs on a timescale of order $\|H_\inter \|^{-1}$, such as the resonance phenomenon discussed below, cannot be described under the BMA.

\textit{Quantum error suppression}.---%
One of the main applications of these bounds is in the context of quantum error suppression. $\mc{C}$ is then the \emph{code subspace} and one is usually interested in the case where it is a degenerate eigensubspace of $H_\sys$ (i.e., $\delta E = 0$). In this case $\rho_\id(t)=\rho(0)$, whence $F\!\left[\rho(t),\rho_\id(t)\right]$ is simply the fidelity between the initial state and the state at time $t$. 
The fidelity can degrade even if the gap is large compared to the interaction, i.e., if $\|H_\inter\|/ \Delta E\le \epsilon\ll1$. In this limit bound~\eqref{main_bound_open} implies that the rate of fidelity loss  is upper bounded by $\Omega'(t)=2\sqrt{2}\Omega(t)+\sqrt{2}\,\IS(P_0 H_\inter P_0)+\mathcal{O}(\epsilon)\|H_\inter\|$. This result has a simple interpretation: fidelity loss can happen either because of leakage, whose speed is bounded by $\Omega(t)$, or because of the effect of the reservoir on $\mathcal{C}$, whose strength is given by $\IS(P_0 H_\inter P_0)$.  In the limit $\Delta E \rightarrow \infty$  the rate of fidelity loss is determined just by the induced splitting $\IS(P_0 H_\inter P_0)$, and if this quantity vanishes then for any finite time $t$, $F\!\left[\rho(t),\rho_\id(t)\right]\rightarrow 1$. 

Therefore, the special case where $\IS(P_0 H_\inter P_0)=0$ is particularly important for error suppression. To illuminate it, consider the decomposition of the interaction term as $H_\inter =\sum_\alpha S^\alpha\otimes B^\alpha$, where $\{S^\alpha\}$ and $\{B^\alpha\}$ are, respectively, independent system and reservoir operators. Then, using Eq.~\eqref{IS:Def}, we find that  
$\IS(P_0 H_\inter P_0)=0$ iff $P_\mc{C} S^\alpha P_\mc{C} \propto P_\mc{C}$ for all $S^\alpha$. This is also known as the quantum error detection (QED) condition \cite{Knill:1997kx,Lidar-Brun:book}. Thus the induced splitting quantifies the deviation from the QED conditions. $\IS(P_0 H_\inter P_0)=0$ can be the result of symmetries of the interaction as in a decoherence-free subspace \cite{Zanardi:97c,Lidar:1998fk}, or it can be engineered using QED codes (e.g., \cite{PhysRevA.74.052322}).  Using bound \eqref{main_bound_open} we can go beyond this special case and study the effectiveness of a particular error suppressing scheme in the case where the perfect QED condition does not hold (see also Ref.~\cite{Marvian:2014nr}).

\textit{Role of the reservoir and system parameters.}---%
One of the interesting aspects of bounds~\eqref{bound_prob1}-\eqref{main_bound_open} is that they are independent of the reservoir state, and the reservoir Hamiltonian enters only via $\|[H_\inter ,H_\B(t)]\|$. This means   that even if the reservoir is infinitely large and $\|H_\B(t)\|$ or $\|d H_\B(t)/dt\|$ are unbounded, as long as  ${\|[H_\inter ,H_\B]\|}$ remains small and bounded, the leakage can be a slow process, depending on the ratio $\Delta E/\|[H_\inter ,H_\B]\|$. This happens, in particular, when the interaction with the reservoir is quasilocal, i.e., the system degrees of freedom (DOFs) interact only weakly with the distant reservoir DOFs. To be concrete, consider the decompositions $H_\inter = \sum_{i\in
\textrm{R}} H_{\inter}^{(i)}$, where each term $H_{\inter}^{(i)}$
acts non-trivially only on a local DOF $i$ in the reservoir. Then
$\|[H_\inter ,H_\B(t)]\|  \le 2 \sum_{i\in \textrm{R}} \|
H_{\inter}^{(i)}\|\, \|H^{(i)}_\B(t)\|$ where $H^{(i)}_\B(t)$ is the
sum of all the terms in $H_\B(t)$ which act non-trivially on DOF $i$.
In many physical scenarios this sum, and hence $\|[H_\inter ,H_\B(t)]\| $, is bounded and small while $\|H_\B(t)\|$ is unbounded and contains long-range interactions. E.g., the reservoir may contain bosonic DOFs, for which $\|H_\B(t)\|$ is infinite. But, as long as these bosonic DOFs do not directly interact with either the system or the DOFs which directly couple to the system (i.e., those with $H_{\inter}^{(i)}\neq 0$), $\|[H_\inter ,H_\B(t)]\|$ can be small. This remains true even if information propagates arbitrarily fast through the reservoir and the Lieb-Robinson velocity \cite{LiebRobinson} is unbounded. 

On the other hand, if $\|[H_\inter ,H_\B(t)]\|$ is large and comparable to $\Delta E \|H_\inter\|$, then our QSL bounds suggest that the timescales for fidelity loss and leakage error can be as small as $\|H_\inter\|^{-1}$, even in the large $\Delta E$ limit. As we explicitly show below, the bounds are attainable when $\|[H_\inter ,H_\B(t)]\|\simeq\Delta E \|H_\inter\|$. 

It is also interesting to note that  bounds~\eqref{bound_prob1}-\eqref{main_bound_open} are independent of the state of the reservoir. This implies that even in the limit of infinitely high temperature $T$, leakage can still be a very slow process, depending on the ratio $\Delta E/{\|[H_\inter ,H_\B]\|}$. This is  a consequence of the assumption that both $\|H_\inter\|$ and $\|[H_\inter ,H_\B]\|$ are bounded, and it does not hold, e.g., in the case of bosonic reservoirs with the standard spin-boson coupling to the system.  On the other hand, even at $T=0$ all information in the system state can be erased by the reservoir in a time of order $\|H_\inter \|^{-1}$, the shortest timescale over which the reservoir can have any influence on the system. Thus, the time it takes the reservoir to affect the evolution of the system is not necessarily related to $T$.

\textit{A model of resonance.}---%
To study the dependence of the time scales for leakage and fidelity loss on the system and reservoir parameters, we present an illustrative example. This is a simple model that exhibits the phenomenon of resonance between the system and reservoir, and is relevant, e.g., also in the context of state transfer via spin chains \cite{Bose:2003uq}. The system is a single qubit ($k=1$) with Hamiltonian $\Delta E_1 {\sigma}_1^z/2$, and gap $\Delta E_1$. The reservoir can have an arbitrarily large number of DOFs and may contain bosonic modes. The only assumptions we make about the structure of the reservoir are (i) the only reservoir DOF which directly couples to the system is another qubit ($k=2$), and (ii) the interaction between the reservoir qubit and other reservoir DOFs, denoted by $h_{2,\textrm{rest}}(t)$, is bounded.  The total Hamiltonian is
\begin{equation}
H_\tot(t)=\sum_{k=1}^{2} \frac{\Delta E_k}{2} {\sigma}^z_k+ J {\vec{\sigma}_1} \cdot \vec{\sigma}_{2}+h_{2,\textrm{rest}}(t)+H_\textrm{rest}(t)\ ,
\end{equation}
where $H_\textrm{rest}(t)$ is an arbitrary Hamiltonian that acts trivially on qubits $1$ and $2$. 

The system qubit is initially in a ${\sigma}_1^z$ eigenstate, and the reservoir is in an arbitrary initial state. 
It turns out that the system's behavior differs strongly between the \emph{resonance} ($|\Delta E_1-\Delta E_2|\ll |J|$) and \emph{out-of-resonance} ($|\Delta E_1-\Delta E_2|\gg |J|$) regimes. To demonstrate this it is useful to transform to the rotating frame defined by $\ket{\phi}\mapsto \exp[i \Delta E_2 t({\sigma}^z_1+{\sigma}^z_2)]\ket{\phi}$. Both the leakage probability of the system qubit and the Heisenberg Hamiltonian are invariant under this unitary transformation. Thus, the new total Hamiltonian in the rotating frame can be obtained from $H_\tot(t)$ by replacing $\Delta E_k\mapsto \Delta E_k-\Delta E_2$, $k=1,2$ and $h_\text{2,rest}(t)\mapsto  \exp[i \Delta E_2 {\sigma}^z_2] h_\text{2,rest}(t) \exp[-i \Delta E_2 {\sigma}^z_2] $. Therefore, the system's energy gap in this rotating frame is $\Delta E_1-\Delta E_2$. 

In the resonance regime leakage can occur in a time of $\mathcal{O}(|J|^{-1})$, the fastest time allowed by the fundamental QSL bound \eqref{general_bound}.  This happens, e.g., already in the case of single qubit reservoir, i.e., $h_{2,\text{rest}}(t)=H_\text{rest}(t)=0$, for which $H_\tot$ can easily be diagonalized. Under the resonance condition the states of the system and reservoir qubits are then swapped in a time of  $\mathcal{O}(|J|^{-1})$, so the fidelity with the initial state is lost.
On the other hand, using our QSL bound~\eqref{eq:tau_leak}, we find that to have leakage with probability of $\mathcal{O}(1)$ in the out-of-resonance regime, the minimum required time is lower bounded as 
\beq
\tau_\leak\ge c {|J|^{-1}}\, \max\{1, \frac{|\Delta E_1-\Delta E_2|}{ \max_t \|h_{2,\text{rest}}(t) \|}\}\ ,
\label{eq:tau_leak_example}
\eeq
representing a potentially drastic increase in the time required for leakage relative to the minimum time $ c|J|^{-1}$ (where $c$ is a constant of order one) obtained from more standard QSL bounds in the form of Eq.~\eqref{general_bound}.   

This model has several interesting general implications: (i) Increasing the system gap can \emph{increase} fidelity loss because the system may become resonant with reservoir DOFs; 
(ii) The relevant parameter which determines the speed of leakage and fidelity loss is not the system gap but the energy mismatch between the system DOFs and the local reservoir DOFs, i.e., those that couple \emph{directly} to the system. If they are in resonance with the system gap, then the reservoir can be insensitive to the gap, and leakage can happen in a time of order $\tau_{\min}\sim \|H_\inter\|^{-1}$, i.e., as fast as allowed by the fundamental QSL bound~\eqref{general_bound}. On the other hand, if this energy mismatch is large then relaxation is slow, even if the remote DOFs of the reservoir are in resonance with the system. (iii) Our QSL bounds are attainable in the regime where $\|[H_\inter,H_\B]\| \sim \Delta E\|H_\inter\|$. (iv) Applying these bounds in different rotating frames can lead to different independent constraints.

\textit{Beyond the $\Delta E> 2\|H_\inter \|$ assumption}.---%
Finally, we discuss how the large-gap assumption $\Delta E> 2\|H_\inter \|$, used in deriving our previous bounds, can be relaxed.  The key idea is to transform to a rotating frame in which $\Delta E$ becomes larger. Let $Q^{+}_\mc{C}$ ($Q^{-}_\mc{C}$) be the projector onto the subspace spanned by the eigenstates of $H_\sys$ whose eigenvalues are greater (less) than those in $\mc{I}$, and transform to the frame defined by $\ket{\Phi} \mapsto e^{-i t F (Q^{+}_\mc{C}- Q^{-}_\mc{C})}\ket{\Phi}$, where $F\in\mathbb{R}$. As we prove in the SM \cite{SM_CO}, 
the new gap between $\mc{C}$ and $\mc{C}^\perp$ becomes $\Delta E + F$. Consequently, bounds \eqref{bound_prob1}-\eqref{eq:tau_fid} all remain valid for any $ F>2\|H_\inter \|-\Delta E$, after the substitutions
\beq
\label{change-RF}
\Delta E \mapsto \Delta E+F\ ,\ \ \ \ \ \   H_\B(t)\mapsto H_\B(t) - F (Q^{+}_\mc{C}- Q^{-}_\mc{C})\ .
\eeq
Moreover, using this generalization, we find that in the large $F$ limit, bound \eqref{main_bound_open} implies that $\sin \frac{\Theta(t)}{2}\le {9 \| H_\inter  \| t}$, 
which is the same as bound~\eqref{general_bound}, up to a constant. Thus, by varying $F$ from $0$ to $\infty$ we can find a family of bounds which gradually changes from \eqref{main_bound_open} to \eqref{general_bound}, and find the strongest bound for fixed given values of the parameters.

\textit{Conclusions.}---%
In this work we introduced state-independent QSLs on leakage and fidelity loss in a Hamiltonian open system framework. 
The reservoir Hamiltonian $H_\B(t)$ only enters our bounds via $\| [H_\inter,H_\B(t)] \|$, implying that only local reservoir modes play a role in our QSLs.
Another important conclusion concerns the common claim that increasing the system's energy gap $\Delta E$ always results in better protection from coupling to the reservoir. The intuitive basis for this claim is the idea that a large gap suppresses thermal excitations by the Boltzmann factor $e^{-\Delta E/kT}$. Under the BMA, the claim can be justified provided the spectral density of the reservoir is monotonically decreasing~\cite{Albash:2015nx}. However, this condition is often violated, e.g., as in the case of an Ohmic bath.
Our results, which are derived without approximations, demonstrate that this intuition is not always correct. 
Increasing $\Delta E$ can result in a resonance with the reservoir, causing the fidelity to drop on a timescale independent of $\Delta E$, even if $T=0$ and the reservoir is in a pure state. These results demonstrate the utility of state-independent QSL bounds for open system dynamics, and raise new questions about the efficacy of energy gaps in protecting quantum information.

\begin{acknowledgments}
This work was supported under grants ARO W911NF-12-1-0541 and NSF CCF-1254119. We thank P. Zanardi, S. Lloyd and E. Farhi for useful discussions.  
\end{acknowledgments}

\bibliography{refs}

\begin{thebibliography}{54}
\bibitem [{com({\natexlab{b}})}]{comment-anti-Zeno}%
  \BibitemOpen
  \href@noop {} {} ({\natexlab{b}}),\ \bibinfo {note} {{This is conceptually
  similar to the phenomenon of accelerated decoherence via the inverse-Zeno
  effect \cite{Kofman:00,Facchi:05}.}}\BibitemShut {Stop}%
\end{thebibliography}

\newpage

\onecolumngrid

\newpage

\begin{center}
\Large{\textbf{Supplementary Material}}
\end{center}

\appendix


\section{Closed system-Time dependent perturbations}

In the following theorem we find bounds on the effect of a time-dependent perturbation on the evolution of a closed system. This theorem is the main tool we use to prove bounds \eqref{bound_prob1} and \eqref{main_bound_open}. The theorem has several other interesting applications which we shall explore in future work. The proof is given in two parts, in subsections~\ref{app:Thm:time:closed-proof} and \ref{app:Thm:time:closed-proof-part2}.

\begin{mytheorem}
\label{Thm:time:closed}
Suppose  the eigenvalues of  the Hermitian operator ${H}_{0}$ which are in the interval  $\mc{I}\subseteq\mathbb{R}$ are separated from the rest of the spectrum of ${H}_{0}$ by at least $\Delta E$.
Suppose $P_0$ is the projector to the subspace spanned by the eigenstates of $H_0$ whose corresponding eigenvalues are in $\mc{I}$, and let $Q_0 = I-P_0$.  Suppose $V(t)$ is a time dependent perturbation that is turned on at $t=0$ (i.e., $V(t)=0$ for $t<0$) and which for all $t>0$ is i) differentiable, ii) satisfies $2\|V(t)\|< \Delta E$. Let $U(t)$ be the unitary evolution  generated by the Hamiltonian $H(t)\equiv H_0+V(t)$, i.e., $U(0)=I$ and $d U(t)/dt=-i H(t) U(t)$ . Then,
  \begin{align}\label{bound-main-time2}
\|Q_0{U}(t)P_0\|&\le {2}\frac{ \|V(0^+)\|+\|V(t)\| }{\Delta E}+  \int_{0^+}^{|t|} d\tau \frac{ {2} \|\dot{V}(\tau) \|}{\Delta E-2\|V(\tau)\|}\ ,
\end{align}
where $\|V(0^+)\|$ is the norm of perturbation at $t=0^+$.   

Furthermore, let $\delta E$ be the difference between the minimum and the maximum eigenvalues of $H_0$ in $\mc{I}$.  Let $U_{\infty}(t)$ be the the unitary generated by the Hamiltonian $ H_\infty(t)\equiv P_0H(t)P_0$ (or, equivalently, the unitary generated by $H'_\infty(t)\equiv H_0+P_0 V(t)P_0$), such that $U_{\infty}(0)=I$,  $d U_{\infty} (t)/dt=-i H_{\infty} (t) U_{\infty}(t)$ (or  $d U_{\infty} (t)/dt=-i H'_{\infty} (t) U_{\infty}(t)$). Then, the following bound holds 
\begin{align}
\label{bound-main-time}
&\|{U}_\infty(t)P_0-U(t)P_0\|\le  {4}\left(\frac{\|V(0^+)\|}{\Delta E}+ \int_{0^+}^{|t|} d\tau\  \left[ \frac{ \|V(\tau)\| \left(\delta E+\|V(\tau)\|\right)}{\Delta E} +\frac{\|\dot{V}(\tau) \|}{\Delta E-2\|V(\tau)\|}\right]\right)\ .
\end{align}
\end{mytheorem}

Note that the Hamiltonian $H'_\infty(t) \equiv H_0+P_0 V(t) P_0$ is the Hamiltonian obtained from first order perturbation theory in the limit $\Delta E\rightarrow\infty$. Bound~\eqref{bound-main-time} compares the dynamics  generated by the actual Hamiltonian $H(t)$ and the effective Hamiltonian $H'_\infty(t)$, given that the initial state is in the support of $P_0$; if the right-hand side is small then we know that for states which are initially in the support of $P_0$ the final states under these two Hamiltonians are almost the same. The following Corollary highlights this point.

\begin{corollary} 
\label{cor:1}
Let $\ket{\psi}$ be an arbitrary state in the support of $P_0$. Let $\ket{\psi_\infty(t)}=U_\infty(t)\ket{\psi}$ and $\ket{\psi(t)}=U(t)\ket{\psi}$ be the states evolving under the evolutions generated by $H_\infty(t)$ and $H(t)$, respectively. Then the fidelity $F(t) = \left|\bracket{\psi(t)}{\psi_\infty(t)}\right|$ between $\ket{\psi_\infty(t)}$ and $\ket{\psi(t)}$ satisfies the following bound:
\begin{align}
&\sqrt{1-F(t)}\le \frac{ {4}}{\sqrt{2}}\left(\frac{\|V(0^+)\|}{\Delta E}+ \int_{0^+}^{|t|} d\tau\  \left[ \frac{ \|V(\tau)\| \left(\delta E+\|V(\tau)\|\right)}{\Delta E} +\frac{\|\dot{V}(\tau) \|}{\Delta E-2\|V(\tau)\|}\right]\right).
\end{align}

\end{corollary}

\begin{proof}
This Corollary follows directly from Theorem~\ref{Thm:time:closed} together with the fact that
\begin{align}
\label{opertor-norm}
\sqrt{1-F(t)}&\le \frac{\|\ket{\psi(t)}-\ket{\psi_\infty(t)}\|}{\sqrt{2}}=  \frac{\|[U(t)-U_{\infty}(t)]\ket{\psi}\|}{\sqrt{2}}\nonumber\\  
& \leq \max_{\ket{\phi} \in \textrm{supp}(P_0)} \frac{\|[U(t)-U_{\infty}(t)]\ket{\phi}\|}{\sqrt{2}}
= \frac{\|[U(t)-U_{\infty}(t)]P_0\|}{\sqrt{2}}\ ,
\end{align}
where in the first inequality we used $\sqrt{1-|\bracket{\phi_1}{\phi_2}|}\le (\|\ket{\phi_1}-\ket{\phi_2}\|)/\sqrt{2}$, and in the last equality we used the definition of the operator norm.
\end{proof}

\section{From time-dependent closed systems to open systems (Proof of bounds~\eqref{bound_prob1}, \eqref{main_bound_open} and \eqref{general_bound})}

\subsection{Proof idea}

To prove our bounds we consider the system and reservoir together as a closed system, and we transform to the rotating frame defined by the transformation $|\phi(t)\rangle\mapsto |\hat{\phi}(t)\rangle=Z_\B(t) |\phi(t)\rangle$,  where $Z_\B(t)$ denotes the unitary generated by the reservoir  Hamiltonian $H_\B(t)$, i.e. 
\begin{equation}\label{def:Z_B}
\dot{Z}_\B(t)=i Z_\B(t) H_\B(t), \ \  Z_\B(0)=I\ .
\end{equation}
So, $Z_\B(t)=e^{i H_\B t}$ in the special case where $H_\B(t)=H_\B$, i.e., is time-independent. 

In this rotating frame the evolution of  system and reservoir is governed by the Hamiltonian $H_\sys+\hat{H}_\inter(t)$, where 
\beq
\hat{H}_\inter(t)= Z_\B(t) H_\inter  Z_\B^\dagger(t)\ .
\eeq 
Note that the reduced state of the system does not change under this change of reference frame. In the rotating frame all the relevant information about the reservoir Hamiltonian and the interaction of the system with the reservoir is encoded in the term $\hat{H}_\inter(t)$. Thus, e.g., the difference between a reservoir with a ``strong'' Hamiltonian $H_\B(t)$ and a reservoir with a ``weak'' Hamiltonian $k H_\B(t)$ with $k\ll 1$, is that in the case of the strong Hamiltonian $\hat{H}_\inter(t)$ oscillates much faster. Furthermore, any property of the reservoir Hamiltonian that does not show up in $\hat{H}_\inter(t)$ is irrelevant from the point of view of the system, i.e., the reduced state of the system remains independent of such a property. 

To prove bounds \eqref{bound_prob1} and \eqref{main_bound_open} we treat  the interaction term $\hat{H}_\inter(t)$ as a time-dependent perturbation of the system Hamiltonian $H_\sys$. In particular, to prove bound \eqref{main_bound_open}
 we  approximate the evolution generated by the Hamiltonian  $H_\sys+\hat{H}_\inter(t)$ with the evolution generated by  the Hamiltonian obtained from first order perturbation theory, i.e., 
\begin{align}
\label{app_Ham}
H_1(t) \equiv H_\sys+\left(P_\mc{C}\otimes I_\B\right) \hat{H}_\inter(t) \left(P_\mc{C}\otimes I_\B\right)\ .
\end{align}
Then we use Theorem~\ref{Thm:time:closed} to find an upper bound on the error we make because of this approximation. In particular, in the limit where the gap $\Delta E \rightarrow \infty$ the approximation error vanishes. Indeed, with the exception of the term proportional to $\IS(P_0 H_\inter P_0)$, all the terms in bound~\eqref{main_bound_open} arise from this approximation, and they appear because we use $H_1(t)$ instead of the actual Hamiltonian $H_\sys+\hat{H}_\inter(t)$.  
The origin of the term involving $\IS(P_0 H_\inter P_0)$ in bound~\eqref{main_bound_open} is the second term in  $H_1$, and it basically describes how the original evolution of the system, which is generated by $H_\sys$, is affected by the second term in $H_1(t)$.  In particular, if $\IS(P_0 H_\inter P_0)=0$ then during the evolution generated by $H_1(t)$ the system is unaffected by the reservoir.

\subsection{Proof of bound~\eqref{bound_prob1}}
\label{app:eq2-proof}
\begin{proof}
First we move to the rotating frame defined by the transformation 
\bes
\begin{align}
\ket{{\Psi}(t)}&\mapsto \ket{\hat{\Psi}(t)}=Z_\B(t) \ket{{\Psi}(t)}\\
O(t)&\mapsto \hat{O}(t)=Z_\B(t) O(t) Z^\dag_\B(t)\ ,
\end{align}
\ees
and the hat denotes states and operators in the rotating frame. 
The system and reservoir together are a closed system and as a joint system they evolve unitarily.  In the rotating frame $\ket{\hat{\Psi}(t)}$, the joint state of the system and reservoir, evolves according to the Schr\"odinger equation 
\begin{align}
\frac{d}{dt} \ket{\hat{\Psi}(t)}=-i (H_\sys+\hat{H}_\inter (t))\ket{\hat{\Psi}(t)}\ .
\end{align}
The main idea is to treat $\hat{H}_\inter (t)$ as a time-dependent perturbation on $H_\sys$ and to use the first part of Theorem~\ref{Thm:time:closed}. Let us write bound \eqref{bound-main-time2} as
  \begin{align}
  \label{bound-main-time2ap}
\|Q_0{S}(t)P_0\|\le  {2} \frac{\|V(0^+)\|+\|V(t)\|}{\Delta E} +  \int_{0^+}^{|t|} d\tau \frac{ {2}\|\dot{V}(\tau) \|}{\Delta E-2\|V(\tau)\|}\ ,
\end{align}
where $S(t)$ is the unitary generated by $H_0+V(t)$, such that $dS(t)/dt =-i(H_0+V(t))S(t)$\ . 

To apply this bound we identify $H_\sys\otimes I_\B$ as the initial Hamiltonian in Theorem~\ref{Thm:time:closed}, i.e., 
\beq
H_0=H_\sys\otimes I_\B\ ,\ \ \ \ \ \ \ \ \ \ \ P_0=P_\mc{C}\otimes I_\B\ \ \ \ \ \ \ \text{and}\ \ \ \ \ \ \ V(t)=\hat{H}_\inter (t)\ .
\eeq
By definition $S(t)$ is then the unitary generated by $H_\sys\otimes I_\B+\hat{H}_\inter (t)$, i.e., it is the unitary that describes the evolution of the system and reservoir in the rotating frame. This means that  $S(t)=Z_\B(t) U(t)$, where $U(t)$ is the unitary describing the evolution of system in the lab frame, i.e., it is the solution of $\dot{U}=-i [H_\sys+H_\B(t)+H_\inter,U(t)]$, $U(0)=I$.  Then, using the fact that $H_\B(t)$ commutes with $P_\mc{C}$ it follows that $Z_\B(t)$ commutes with $P_0$, and so the left-hand side of bound~\eqref{bound-main-time2ap} is 
\begin{align}
\textrm{LHS}=\|Q_0Z_\B(t) U(t) P_0\|=\|Z_\B(t) Q_0 U(t) P_0\|= \|Q_0 U(t) P_0\|\ .
\end{align}
On the other hand, using the facts that $\| \hat{H}_\inter (t)\|=\| {H}_\inter (t)\|$ and 
$\left\| \frac{d}{dt}  \hat{H}_\inter (t) \right\| =\left\| Z_\B(t)  [ H_\B(t), {H}_\inter  ] Z^\dag_\B(t)  \right\|=\left\|  [ H_\B(t), {H}_\inter  ]  \right\|$,
we find that the right-hand side of bound~\eqref{bound-main-time2ap} is
 \begin{align}
\textrm{RHS}= \frac{ {4} \|H_\inter  \|}{\Delta E} +  \int_{0^+}^{|t|} d\tau \frac{ {2} \| [ H_\B(t), H_\inter  ] \|}{\Delta E-2\|H_\inter \|} \ .
\end{align}
Therefore, bound~\eqref{bound-main-time2ap} implies
 \begin{align}\label{bbb2}
\|Q_0 U(t) P_0\|\le\frac{ {4} \|H_\inter  \|}{\Delta E} +  \int_{0^+}^{|t|} d\tau \frac{ {2} \| [ H_\B(t), H_\inter  ] \|}{\Delta E-2\|H_\inter \|} \ .
\end{align}

To prove bound~\eqref{bound_prob1} we assume  that the support of the initial reduced density operator of the system is restricted to $\mc{C}$, i.e.,  
\beq
P_0\ \sigma_\textrm{SR}(0) P_0= \left(P_\mc{C}\otimes {I}_\B\right)\sigma_\textrm{SR}(0) \left(P_\mc{C}\otimes {I}_\B\right)= \sigma_\textrm{SR}(0) \ ,
\eeq
where $\sigma_\textrm{SR}(0)$ is the initial joint system-reservoir state. Recall that the leakage probability is $p_\leak(t)\equiv\Tr\!\left[ Q_\mc{C} \rho(t) \right]$. It follows from bound~\eqref{bbb2} that 
\bes
\begin{align}
p_\leak(t)&=\Tr \left[ Q_0 U(t)  \sigma_\textrm{SR}(0) U^\dag(t) \right]=\Tr \left[ Q_0 U(t)  P_0 \sigma_\textrm{SR}(0) P_0 U^\dag(t) Q_0 \right] \\ 
\label{eq:F10c}
&\le \|Q_0 U(t)  P_0\|^2 \|\sigma_\textrm{SR}(0)\|_1\\ 
&\le   \left(\frac{ {4} \|H_\inter  \|}{\Delta E} +  \int_{0^+}^{|t|} d\tau \frac{ {2} \| [ H_\B(t), H_\inter  ] \|}{\Delta E-2\|H_\inter \|}\right)^2 \ ,
\end{align}
\ees
where to obtain bound~\eqref{eq:F10c} we used the inequality $|\Tr(AB)| \leq \|A\| \|B\|_1$, valid for any two operators $A$ and $B$ \cite{Bhatia:book}.
\end{proof}


\subsection{Proof of bound \eqref{main_bound_open}}
\label{app:eq4-proof}
\begin{proof}
We shall use the second part of Theorem~\ref{Thm:time:closed}. We consider the system and reservoir together as a closed system. Without loss of generality we assume that the state of the joint system is pure. If this is not the case then we can always purify the joint state by adding an auxiliary system which has a trivial Hamiltonian and is not interacting with the system and the reservoir.  

Let $\ket{\Phi(0)}$ be  the joint state of system and reservoir at $t=0$, which need not be factorizable. Then, the joint state of system and reservoir at time $t$ is  given by $\ket{\Phi(t)}=U(t) \ket{\Phi(0)}$, where $U(t)$ is the joint unitary evolution of the system and reservoir, and is generated by $H_\tot=H_\sys+H_\B(t)+H_\inter$, i.e.  $d{U}(t)/dt=-iH_\text{tot}(t)U(t)$, and $U(0)=I$ . Thus $\ket{\Phi(t)}$ is a purification of $\rho(t)$. On the other hand,  $e^{- i t H_\sys}\ket{\Phi(0)}$ is a purification of  ${\rho_\id(t)}$, i.e.,
\beq
\Tr_\B\left(e^{- i t H_\sys}\ket{\Phi(0)} \langle\Phi(0)| e^{i t H_\sys}  \right)=\rho_\id(t)\ .
\eeq
Next, consider the Hamiltonian 
\beq
H_\dec(t)\equiv H_\sys+H_\B(t)+ P_\mc{C}\otimes K_\B ,
\eeq
where $K_\B\in  \textrm{Herm}(\mc{H}_\B)$ is an arbitrary Hermitian operator acting on the reservoir Hilbert space, which will be determined later [below Eq.~\eqref{bound_FD}]. Note that this Hamiltonian can be obtained from the original Hamiltonian of the system and reservoir  by replacing $H_\inter $ with $ P_\mc{C}\otimes K_\B$. Also, note that from the point of view of states which are in the code subspace $P_\mc{C}$, under evolution generated by this Hamiltonian the system and reservoir are decoupled from each other:  Let $U_\dec(t) $ be the unitary evolution generated by $H_\dec(t)$, such that $d{U}_\dec(t)/dt=-iH_\dec(t)U_\dec(t)$, and $U_\dec(0)=I$ . Then, using the fact that  $P_\mc{C}$ commutes with $H_\sys$, and $H_\B(t)+ P_\mc{C}\otimes K_\B$ acts trivially on  $|\Phi(0)\rangle$ we find
\begin{align}
\Tr_\B\left(U_\dec(t) \ket{\Phi(0)} \langle\Phi(0)| U^\dag_\dec(t)   \right)=
\Tr_\B\left(e^{- i t H_\sys}\ket{\Phi(0)} \langle\Phi(0)| e^{i t H_\sys}  \right)=\rho_\id(t)
\end{align}
In other words, $U_\dec(t)\ket{\Phi(0)}$ is another purification of $\rho_\id(t)$.

Since $U(t) \ket{\Phi(0)}$ and  $U_\dec(t)\ket{\Phi(0)}$ are, respectively, purifications of $\rho(t)$ and $\rho_\id(t)$, we can use Uhlmann's theorem \cite{nielsen2000quantum}: 
\beq
F[\rho(t),\rho_\id(t)]\ge \left|  \bra{\Phi(0)} U^\dag_\dec(t) U(t) \ket{\Phi(0)}  \right| .  
\eeq
Define 
\beq
\label{H-def:inf}
H_\infty(t) \equiv H_\sys+H_\B(t)+P_0 H_\inter  P_0\ ,\qquad P_0=P_\mc{C}\otimes I_\B ,
\eeq
and let $U_\infty(t)$ be the unitary generated by $H_\infty(t)$, i.e.
\beq
 \frac{dU_\infty(t)}{dt}=-i t H_\infty(t) U_\infty(t), \qquad  U_\infty(0)=I  .
\eeq
Then,  using the inequality $\sqrt{1-|\bracket{\phi_1}{\phi_2}|}\le (\|\ket{\phi_1}-\ket{\phi_2}\|)/\sqrt{2}$ and the triangle inequality we then find
\bes
\label{Fidelity1}
\begin{align}  
\sqrt{1-F[\rho(t),\rho_\id(t)]}&\le \sqrt{1- \left|  \bra{\Phi(0)} U^\dag_\dec(t) U(t) \ket{\Phi(0)}  \right| }  \le  \frac{\|U(t) \ket{\Phi(0)}  -U_\dec(t) \ket{\Phi(0)}  \| }{\sqrt{2}}\\ 
\label{eq:F20b}
&\le  \frac{\|U(t)P_0-U_\dec(t)P_0 \| }{\sqrt{2}} \leq \frac{\|U(t)P_0-U_\infty(t)P_0\|+\|U_\infty(t)P_0-U_\dec(t)P_0\|}{\sqrt{2}}\\
&\le   \frac{\|U(t)P_0-U_\infty(t)P_0\|+\|U_\infty(t)-U_\dec(t)\|}{\sqrt{2}}\ ,
\end{align}
\ees
where in inequality~\eqref{eq:F20b} we used the definition of the operator norm (this is essentially identical to the proof of Corollary~\ref{cor:1}).

Using Lemma~\ref{lemma3} we have 
\begin{align}
\label{bound_FD}
\|U_\infty(t)-U_\dec(t)\| \le \int_{0^+}^{|t|} \ ds \|H_\infty(s)-H_\dec(s)\| =
|t|\, \|P_0 H_\inter P_0-P_\mc{C}\otimes K_\B \|\ .
\end{align}
Now we choose  $K_\B$ to be the Hermitian operator for which $\|P_0 H_\inter P_0-P_\mc{C}\otimes K_\B \|$ is minimized.
Bound~\eqref{bound_FD} thus implies
\begin{align}
\label{bound_FD2}
\|U_\infty(t)-U_\dec(t)\|\le |t|\, \min_{K_\B\in  \textrm{Herm}(\mc{H}_\B)} \|P_0 H_\inter P_0-P_\mc{C}\otimes K_\B \| \le  |t|\ \IS(P_0 H_\inter P_0)\ ,
\end{align}
where we used definition~\eqref{IS:Def}. As we mentioned in the main text, note that $\IS(P_0 H_\inter P_0)$ can be nonzero only when $\mc{C}$ has dimension larger than one. To see this note that if $\dim\mc{C}=1$ then $P_\mc{C} = \ketbra{\psi}{\psi}$ for a normalized state $|\psi\rangle$; then, choosing  $K_\B=\langle\psi|H_\inter|\psi\rangle$ we find  
$(|\psi \rangle\langle\psi|\otimes I_\B) H_\inter(|\psi \rangle\langle\psi|\otimes I_\B)=|\psi \rangle\langle\psi|\otimes\langle\psi| H_\inter|\psi \rangle$, and so $\IS(P_0 H_\inter P_0)=0$ in this case.

Using bound~\eqref{bound_FD2} in bound~\eqref{Fidelity1}, we find
\begin{align}  
\label{Fidelity3}
\sqrt{1-F[\rho(t),\rho_\id(t)]}&\le \frac{\|U(t)P_0-U_\infty(t)P_0\|}{\sqrt{2}}+ \frac{ |t|}{\sqrt{2}}\IS(P_0 H_\inter P_0) \ .
\end{align}
We show below, using bound~\eqref{bound-main-time} of Theorem~\ref{Thm:time:closed}, that
\beq
\label{Fid_b1}
\|U(t)P_0- U_{\infty}(t)P_0\|\le \frac{ {4}\|H_\inter \|}{\Delta E}+  {4}|t|  \frac{\|H_\inter \|{(\delta E+\|H_\inter \|)}}{{\Delta E}}+ {4} \int_{0^+}^{|t|} d\tau \frac{\|[H_\inter ,H_\B(\tau)]\|}{\Delta E-2 \|H_\inter \|}\ \ .
\eeq
Combining this with bound~\eqref{Fidelity3} proves bound~\eqref{main_bound_open}. It thus remains to prove bound~\eqref{Fid_b1}, which we do next.

The argument is analogous to the argument we used to prove bound~\eqref{bound_prob1} using bound~\eqref{bound-main-time2}. As in Sec.~\ref{app:eq2-proof} we start by transforming to the rotating frame defined by $\ket{{\Psi}(t)}\mapsto \ket{\hat{\Psi}(t)}=Z_\B(t) \ket{{\Psi}(t)}$, where $Z_\B(t)$ is defined in Eq.~\eqref{def:Z_B}. In this rotating frame $\ket{\hat{\Psi}(t)}$ (the joint state of the system and reservoir) evolves according to the Schr\"odinger equation $\frac{d}{d\tau} \ket{\hat{\Psi}(t)}=-i (H_\sys+\hat{H}_\inter (t))\ket{\hat{\Psi}(t)}$, 
where $\hat{H}_\inter (t)=Z_\B(t) {H}_\inter  Z^\dag_\B(t)$.  Again, the  idea is to treat $\hat{H}_\inter (t)$ as a time-dependent perturbation on $H_\sys$ and to use bound~\eqref{bound-main-time} of Theorem~\ref{Thm:time:closed} to estimate the effect of this perturbation. Let us rewrite bound~\eqref{bound-main-time} in the form
\begin{align}
\label{Eq:copy1}
\|{S}_\infty(t)P_0-S(t)P_0\|\le \frac{ {4}\|V(0^+)\|}{\Delta E}+\frac{ {4}}{\Delta E} \int_{0^+}^{|t|} d\tau\   \|V(\tau)\| \left(\delta E+\|V(\tau)\|\right)+  {4}  \int_{0^+}^{|t|} d\tau \frac{\|\dot{V}(\tau) \|}{\Delta E-2\|V(\tau)\|} \ ,
\end{align}
where $S(t)$ and ${S}_\infty(t)$ are, respectively, the unitaries generated by the Hamiltonians $H_0+V(t)$ and $H_0+P_0 V(t) P_0$.
 
To apply bound~\eqref{bound-main-time} we identify $H_\sys\otimes I_\B$ as the initial Hamiltonian $H_0$ in Theorem~\ref{Thm:time:closed}, i.e.,
\beq\label{identific}
H_0=H_\sys\otimes I_\B\ ,\ \ \ \ \ \ \ \ \ \ \ P_0=P_\mc{C}\otimes I_\B\ \ \ \ \ \ \ \text{and}\ \ \ \ \ \ \ V(t)=\hat{H}_\inter (t)\ .
\eeq

First, consider the left-hand side of Eq.~\eqref{Eq:copy1}. Since $S(t)$ is the unitary generated by $H_\sys+\hat{H}_\inter (t)$, then $S(t)=Z_\B(t)U(t)$. Similarly,  ${S}_\infty(t)$ is the unitary generated by 
\beq
H_\sys+(P_\mc{C}\otimes I_\B) \hat{H}_\inter (t) (P_\mc{C}\otimes I_\B) =
H_\sys+ Z_\B(t) \left( P_0 {H}_\inter  P_0 \right)  Z^\dag_\B(t)\ .
\eeq
Going from the rotating frame to the lab frame, one can easily see that this Hamiltonian transforms to  $H_\sys+ H_\B+P_0 {H}_\inter  P_0 $ which, by definition [Eq.~\eqref{H-def:inf}] is $H_\infty$. Since $U_\infty(t)$ is the unitary generated by $H_\infty$, it follows that $S_\infty(t)=Z_\B(t) U_\infty(t)$. Therefore, the left-hand side of Eq.~\eqref{Eq:copy1} is
\begin{align}
\label{LHS5}
\textrm{LHS}=\|S_\infty(t) P_0- S(t) P_0\|=\|Z_\B(t)\left[ {U}(t)P_0- {U}_{\infty}(t)P_0\right]\|=\|{U}(t)P_0- {U}_{\infty}(t)P_0\|\ .
\end{align}
On the other hand, using Eq.~\eqref{identific} we can easily see that the right-hand side of Eq.~\eqref{Eq:copy1} is
\begin{align}
\textrm{RHS}&=\frac{ {4}\|H_\inter \|}{\Delta E}+ \frac{ {4}}{\Delta E}\int_{0^+}^{|t|} d\tau  \|H_\inter \|{(\delta E+\|H_\inter \|)}+ {4} \int_{0^+}^{|t|} d\tau \frac{\|[H_\inter ,H_\B(\tau)]\|}{\Delta E-2 \|H_\inter \|}\\ &= \frac{ {4}\|H_\inter \|}{\Delta E}+  {4}|t|  \frac{\|H_\inter \|{(\delta E+\|H_\inter \|)}}{{\Delta E}}+{4} \int_{0^+}^{|t|} d\tau \frac{\|[H_\inter ,H_\B(\tau)]\|}{\Delta E-2 \|H_\inter \|} \ , \label{RHS5}
\end{align}
where again we used  $\|\hat{H}_\inter (t)\|=\|{H}_\inter \|$, and $\|\frac{d}{d t}\hat{H}_\inter (t)\|=\|[{H}_\inter ,H_\B(t)]\|$.

Substituting Eqs.~\eqref{LHS5} and \eqref{RHS5} into bound~\eqref{Eq:copy1} we find
\begin{align*}
\|U(t)P_0- U_{\infty}(t)P_0\| \le \frac{ {4}\|H_\inter \|}{\Delta E}+  {4}|t|  \frac{\|H_\inter \|{(\delta E+\|H_\inter \|)}}{{\Delta E}}+ {4} \int_{0^+}^{|t|} d\tau \frac{\|[H_\inter ,H_\B(\tau)]\|}{\Delta E-2 \|H_\inter \|}\ ,
\end{align*}
which proves Eq.~\eqref{Fid_b1}, and so completes the proof Eq.~\eqref{main_bound_open}.
\end{proof}

\subsection{Proof of bound~\eqref{general_bound}}
\label{app:eq8-proof}
\begin{proof}
We present the proof for the special case where $H_\B(t)=H_\B$ is time-independent. The time-dependent case follows exactly in the same fashion. 

Let $\ket{\Phi(0)}$ be any arbitrary joint pure state of the system and reservoir at $t=0$. Then, the joint state of the system and reservoir at time $t$ is given by $\ket{\Phi(t)}=e^{-i H_\tot} \ket{\Phi(0)}$, where $H_\tot=H_\sys+H_\B+H_\inter $. So, by definition, $\ket{\Phi(t)}$ is a purification of $\rho(t)$, the reduced state of system at time $t$, i.e.,
\beq
\rho(t)= \Tr_\B\left(\ket{\Phi(t)}\langle\Phi(t)|\right) .
\eeq
On the other hand, one can easily see that $\ket{\Phi_\id(t)}\equiv e^{- i t (H_\sys+H_\B)}\ket{\Phi(0)}$ is a purification of  ${\rho_\id(t)}$, i.e.,
\beq
\Tr_\B\left(\ket{\Phi_\id(t)}\langle\Phi_\id(t)|\right)=e^{- i t H_\sys} \Tr_\B\left(\ket{\Phi(0)} \langle\Phi(0)| \right) e^{i t H_\sys}=e^{- i t H_\sys} \rho(0) e^{i t H_\sys}=  \rho_\id(t)\ .
\eeq
 Therefore, using Uhlmann's theorem we find 
\beq
F[\rho(t),\rho_\id(t)]\ge \left|\langle \Phi(t) \ket{\Phi_\id(t)}\right| .  
\eeq
Next, using the fact that $\sqrt{1-|\langle\phi_1\ket{\phi_2}|}\le (\|\ket{\phi_1}- e^{i\theta} \ket{\phi_2}\|)/\sqrt{2}$, where $\theta$ is an arbitrary phase, we find
\bes
\begin{align}  
\sqrt{1-F[\rho(t),\rho_\id(t)]}&\le \sqrt{1-\left|\langle \Phi(t) \ket{\Phi_\id(t)}\right|}  \le  \frac{\|\ket{\Phi(t)}-e^{-i c t}  \ket{\Phi_\id(t)}\| }{\sqrt{2}}\\
&=\frac{\|e^{-i t H_\tot} \ket{\Phi(0)}- e^{-i c t}  e^{- i t (H_\sys+H_\B)}\ket{\Phi(0)} \| }{\sqrt{2}} \le  \frac{\|e^{-i t H_\tot}- e^{-i c t}  e^{- i t (H_\sys+H_\B)}  \| }{\sqrt{2}}  ,
\end{align}
\ees
where $c$ is an arbitrary real constant. 
Finally, using Lemma~\ref{lemma3} we have
\beq
\|e^{-i t H_\tot}- e^{-i c t} e^{- i t (H_\sys+H_\B)}  \| \le |t| \ \|H_\tot- (H_\sys+H_\B+c I) \|= |t| \ \|H_\inter -c I\| \leq |t|\|H_\inter\| \ ,
\eeq
Thus, $\sqrt{1-F[\rho(t),\rho_\id(t)]}\le |t| \ \|H_\inter -c I\|/\sqrt{2} $ for any $c\in\mathbb{R}$. Let $\lambda_{\max}$ and  $\lambda_{\min}$ be the maximum and minimum eigenvalues of $H_\inter $ respectively. Thus the maximum and minimum eigenvalues of $H_\inter -cI$ are $\lambda_{\max}-c$ and $\lambda_{\min}-c$, respectively. It is easy to see that to minimize $\|H_\inter -c I\|$ we should choose $c=(\lambda_{\max}+\lambda_{\min})/2$, whence $\|H_\inter -c I\|=(\lambda_{\max}-\lambda_{\min})/2$. Therefore 
$\sqrt{1-F[\rho(t),\rho_\id(t)]}/\sqrt{2}\le |t|(\lambda_{\max}-\lambda_{\min})/4 \leq |t|\|H_\inter\| /2$, as claimed.
\end{proof}

\subsection{Proof of transformation~\eqref{change-RF}}
\label{app:gen}

Recall that $\mc{I}$ is an interval of $\mathbb{R}$ which includes at least one eigenvalue of $H_\sys$, $\mc{C}$ is the subspace spanned by the eigenvectors of $H_\sys$ whose eigenvalues are in $\mc{I}$, and $P_\mc{C}$ is the projector onto the subspace $\mc{C}$. Recall that $\Delta E$ denotes the gap between the energy levels of $H_\sys$ inside and outside $\mc{C}$ (i.e., if $\lambda_1$ and $\lambda_2$ are two distinct eigenvalues of $H_\sys$ such that $\lambda_1\in \mc{I}$ but $\lambda_2\notin \mc{I}$, then $|\lambda_1-\lambda_2|>\Delta E$). 

\begin{proof}
Let $Q^{+}_\mc{C}$ ($Q^{-}_\mc{C}$) be the projector onto the subspace spanned by the eigenstates of $H_\sys$ whose eigenvalues are greater (less) than those in $\mc{I}$. This definition implies that 
\beq
Q^{+}_\mc{C}+Q^{-}_\mc{C}+P_\mc{C}=I_\sys \ .
\eeq
To prove bounds~\eqref{bound_prob1}-\eqref{main_bound_open},
we moved to the rotating frame described by the transformation
\begin{align}
\ket{{\Psi}(t)}&\mapsto \ket{\hat{\Psi}(t)}=Z_\B(t) \ket{{\Psi}(t)} ,
\end{align}
where $Z_\B(t)$ is defined in Eq.~\eqref{def:Z_B}.
To prove transformation~\eqref{change-RF}  we move to a new rotating frame defined by
\bes
\begin{align}
\ket{{\Psi}(t)}&\mapsto \ket{\hat{\Psi}(t)}=W_F(t) \ket{{\Psi}(t)} , \\
\label{eq:H3b}
W_F(t)&\equiv \exp(-i t F [Q^{+}_\mc{C}- Q^{-}_\mc{C}])\otimes Z_\B(t) ,  
\end{align}
\ees
(i.e., $W_F(t)$ is generated by $H_W(t) = H_\B(t)- F [Q^{+}_\mc{C}- Q^{-}_\mc{C}]$), where $F$ is an arbitrary real number satisfying
\beq
\label{cond45}
F>2\|H_\inter \|-{\Delta E}\ .
\eeq
We will shortly present  the motivation for this condition (more generally, $F$ can be chosen to be time-dependent, but we shall not consider this here). In the rotating frame $\ket{\hat{\Psi}(t)}$ (the joint state of the system and reservoir) evolves according to the Schr\"odinger equation
\begin{align}
\frac{d}{dt} \ket{\hat{\Psi}(t)}=-i (H_0 +  \hat{H}_\inter (t))\ket{\hat{\Psi}(t)}\ ,
\end{align}
where 
\beq
H_0 \equiv H_\sys+F [Q^{+}_\mc{C}- Q^{-}_\mc{C}]\qquad \textrm{and} \qquad \hat{H}_\inter (t)\equiv W_F(t) {H}_\inter  W^\dag_F(t)\ .
\eeq 
Thus $\ket{\hat{\Psi}(t)} = S(t) \ket{\hat{\Psi}(0)}$, where $S(t)$ is the unitary generated by $H_0+\hat{H}_\inter(t)$. Note that $S(t)$ can be written as $S(t)=W_F(t) U(t)$, where $U(t)$ is the unitary describing the joint evolution of the system and the reservoir in the lab frame, i.e., the solution of $d{U}(t)/dt=-i [H_\sys+H_\B(t)+H_\inter]U(t)$ with $U(0)=I$.

Again, we use the first part of Theorem~\ref{Thm:time:closed}. Let us rewrite bound~\eqref{bound-main-time2} as
\begin{align}
\label{bound-main-time2ap55}
\|Q_0{S}(t)P_0\|\le {2} \frac{\|\hat{H}_\inter(0^+)\|+\|\hat{H}_\inter(t)\|}{\Delta D} +  \int_{0^+}^{|t|} d\tau \frac{{2} \|\partial_\tau{\hat{H}}_\inter(\tau) \|}{\Delta D-2\|\hat{H}_\inter(\tau)\|}\ ,
\end{align}
where $P_0 = P_\mc{C}\otimes {I}_\B$.
Let $\Delta D$ denote the gap of Hamiltonian $H_0$. The conditions of Theorem~\ref{Thm:time:closed} require that $\Delta D>2\|\hat{H}_\inter(\tau)\|$ for all $0\le\tau\le t$, and we next show that this can be satisfied provided we choose $F$ as in condition~\eqref{cond45}.

We claim that the gap $\Delta D$ of $H_0$ is $\Delta D=\Delta E + F$. To see this, first note that by definition $Q^{-}_\mc{C}$ and $Q^{+}_\mc{C}$ are diagonal in the eigenbasis of $H_\sys$, so that $H_0$ is diagonal in the same basis. Adding $F [Q^{+}_\mc{C}- Q^{-}_\mc{C}]$ to $H_\sys$ has the effect of adding $F$ to all the eigenvalues of $H_\sys$ greater than those in $\mc{I}$, and subtracting $F$ from all the eigenvalues of $H_\sys$ less than those in $\mc{I}$, while leaving the eigenvalues in $\mc{I}$ alone. Since the gap between the eigenvalues in $\mc{I}$ and those not in $\mc{I}$ was $\Delta E$ before the addition of $F [Q^{+}_\mc{C}- Q^{-}_\mc{C}]$, it becomes $\Delta E + F$ after this addition, which is then the new gap between $\mc{C}$ and the orthogonal subspace.
Then, since $\|\hat{H}_\inter (t)\|=\|H_\inter \|$,  condition~\eqref{cond45} guarantees that 
\beq
\Delta D=F+\Delta E> 2 \|H_\inter \| \ ,
\eeq
and so we can apply Theorem~\ref{Thm:time:closed} and bound~\eqref{bound-main-time2ap55}.

Using the fact that $H_W(t) = H_\B(t)- F [Q^{+}_\mc{C}- Q^{-}_\mc{C}] $ commutes with $P_\mc{C}$ it follows that $W_F(t)$ commutes with $P_0$, and so the left-hand side of Eq.~\eqref{bound-main-time2ap55} is
\begin{align}
\textrm{LHS}=\|Q_0{S}(t)P_0\|=\|Q_0W_F(t) U(t) P_0\|= \|Q_0 U(t) P_0\|
\end{align}
On the other hand, using the facts that $\| \hat{H}_\inter (t)\|=\| {H}_\inter \|$ and 
\beq
\left\| \partial_t  \hat{H}_\inter (t) \right\| =\left\| W_F(t)[H_\inter,H_W(t)  ] W_F^\dagger(t)  \right\|=\left\| [H_\inter,H_\B(t)-F (Q^{+}_\mc{C}-Q^{-}_\mc{C}) ]  \right\|\ ,
\eeq
we find that the right hand side of bound~\eqref{bound-main-time2ap55} is 
\begin{align}
\textrm{RHS}&= \frac{{4} \|H_\inter  \|}{\Delta E+F} +  \int_{0^+}^{|t|} d\tau \frac{{2} \left\| [H_\inter,H_\B(t)-F (Q^{+}_\mc{C}-Q^{-}_\mc{C}) ] \right\|}{\Delta E+F-2\|H_\inter \|}\ .
\end{align}
Therefore, bound~\eqref{bound-main-time2ap55} implies
\begin{align}
\label{bbb24}
\|Q_0 U(t) P_0\|\le\frac{{4} \|H_\inter  \|}{\Delta E+F} +  \int_{0^+}^{|t|} d\tau \frac{{2} \left\| [H_\inter,H_\B(t)-F (Q^{+}_\mc{C}-Q^{-}_\mc{C}) ] \right\|}{\Delta E+F-2\|H_\inter \|}\ .\ .
\end{align}
Suppose $\sigma_\textrm{SR}(0)$ is the initial joint state of the system and reservoir with the property that the reduced state of the system has support only in $\mc{C}$, i.e.,
\beq
P_0\ \sigma_\textrm{SR}(0) P_0= \left(P_\mc{C}\otimes {I}_\B\right)\  \sigma_\textrm{SR}(0)\  \left(P_\mc{C}\otimes {I}_\B\right)= \sigma_\textrm{SR}(0) \ .
\eeq
Then, bound~\eqref{bbb24} implies 
\bes
\begin{align}
p_\leak(t)&=\Tr \left[U(t)  \sigma_\textrm{SR}(0) U^\dag(t) Q_0 \right]\\ &=\Tr \left[ Q_0 U(t)  P_0 \sigma_\textrm{SR}(0) P_0 U^\dag(t) Q_0 \right] \\ &\le \|Q_0 U(t)  P_0\|^2 \|\sigma_\textrm{SR}(0)\|_1\\ &\le   \left(\frac{{4} \|H_\inter  \|}{\Delta E+F} +  \int_{0^+}^{|t|} d\tau \frac{{2} \left\| [H_\inter,H_\B(t)-F (Q^{+}_\mc{C}-Q^{-}_\mc{C}) ] \right\|}{\Delta E+F-2\|H_\inter \|} \right)^2 \ ,
\end{align}
\ees   
which  is bound~\eqref{bound_prob1} after the replacements 
\bes
\begin{align}
H_\B(t) &\mapsto  H_\B(t)-F (Q^{+}_\mc{C}- Q^{-}_\mc{C})\\
\Delta E &\mapsto \Delta E+F\ .
\end{align} 
\ees
The argument which leads to bound \eqref{main_bound_open} can be repeated in a similar fashion.  
\end{proof}

\section{Proof of theorem \ref{Thm:time:closed}}

\subsection{Preliminaries} 

\subsubsection{Three useful Lemmas}
\begin{mylemma}
\label{lemma3}
Suppose the unitaries $U_1(t)$ and $U_2(t)$ are generated by the Hamiltonians $H_1(t)$ and $H_2(t)$, such that  $U_{1,2}(0)=I$ and $dU_{1,2}(t)/d\tau=-i H_{1,2}(t)U_{1,2}(t)$. Then,
\beq
\| U_1(t)-U_2(t)\| \le    \int_0^{|t|} d\tau\ \|H_1(\tau)-H_2(\tau)\|\ .
\eeq
\end{mylemma}
See Sec.~\ref{app:lemma3-proof} for the proof.

\begin{mylemma}
\label{lem:S}
Let $P(t)$ be a time-dependent projector satisfying $[P(t), H(t)]=0$ for all times $t$, where $H(t)$ is the Hamiltonian generating the unitary $U(t)$, i.e., $\dot{U}(t) = -iH(t) U(t)$ with $U(0)=I$. 
Let $H_\tru(t)=P(t) H(t) P(t)$ be the generator of $U_\tru(t)$, i.e., $dU_\tru(t)/dt=-i H_\tru(t)U_\tru(t)$ with $U_\tru(0)=I\ $.
Assume the time-derivative of  $P(t)$ exists for $t>0\ $. Then
\bes
\begin{align}
\label{bound_dot_Pi}
\|U(t) P({0^+})- P(t)U(t)\| & \leq \int_{0^+}^t \|  \dot{P}(\tau) \| \ d\tau \ , \\
\label{bound_dot_P}
\|U(t)P(0^+)-U_\tru(t)P(0^+)\| & \le 2 \int_{0^+}^{|t|}  \|\dot{P}(\tau)\| d\tau\ .,
\end{align}
\ees
\end{mylemma}
See Sec.~\ref{app:Pdot-proof} for the proof.

\begin{mylemma}
\label{lem:23}
Let $P$ and $\tilde{P}$ be the projectors with the same rank. Then 
\beq
\|P-\tilde{P}\|=\|P\tilde{Q}\|=\|\tilde{P}Q\|\ ,
\eeq
where $Q=I-P$ and $\tilde{Q} = I-\tilde{P}$.
\end{mylemma}
See Sec.~\ref{app:lemma23-proof} for the proof.

\subsubsection{Bounds on the effect of perturbations}

Let $\mc{I}_0\subseteq \mathbb{R}$ be any interval containing one or more eigenvalues of a Hermitian operator $H_0$.  Suppose the eigenvalues of $H_0$ in the interval $\mc{I}_0$ are separated from the rest of the eigenvalues of $H_0$ by at least $\Delta E$. I.e., for any pair of eigenvalues  $\lambda_1$ and $\lambda_2$ of $H_0$, if $\lambda_1\in \mc{I}_0$ and  $\lambda_2\notin \mc{I}_0$, then $|\lambda_1-\lambda_2|\ge \Delta E$.  

\begin{mylemma}
\label{First_lem}
Let $V$ be a Hermitian operator satisfying $0<\|V\|< \Delta E/2$. Let $\mc{I}$ be the interval obtained from $\mc{I}_0$ by adding a margin of $\|V\|$ on the left and the right of $\mc{I}_0$. Let $P_0$ and $P$ be, respectively, the projectors onto the subspaces spanned by eigenstates of $H_0$ in the interval $\mc{I}_0$ and  eigenstates of $H_0+V$ in the interval $\mc{I}$. Then, $P_0$ and $P$ have the same rank and
\beq
\|P-P_0\|\le \frac{{2}\|V\|}{\Delta E}\ .  
\eeq
\end{mylemma}
This is Lemma 3.1 of Ref.~\cite{Bravyi20112793}. (See also Theorem VII.3.2 of Ref.~\cite{Bhatia:book}).  

\begin{mylemma}
\label{main:lemma}
Let $V(t)$ be an arbitrary differentiable Hermitian operator satisfying $0<\|V(t)\| < \Delta E/2$. Let $\mc{I}(t)$ be the interval obtained from $\mc{I}_0$ by adding a margin of $\|V(t)\|$ on the left and the right of $\mc{I}_0$.
Let $P(t)$ be the projector onto the subspace spanned by the eigenvectors of $H_0+V(t)$ whose eigenvalues belong to $\mc{I}(t)$. Then,
\bes
\begin{align}
\label{Eq1:lema2} 
\|\dot{P}(t)\|&\le \frac{ {2} \| P(t) \dot{V}(t) Q(t)\|}{\Delta E- 2\|V(t)\|}\\ 
\label{Eq2:lema2} 
&\le \frac{{2}  \| \dot{V}(t)\|}{\Delta E- 2\|V(t)\|}\ .
\end{align}
\ees

\end{mylemma}
See Sec.~\ref{app:main:lemma-proof} for the proof.
 
\subsection{Proof of the first part of Theorem~\ref{Thm:time:closed} [bound \eqref{bound-main-time2}]}
\label{app:Thm:time:closed-proof}

First, using the triangle inequality, we note that for an arbitrary projector $\Lambda$ it holds that 
\bes
\begin{align}
\|(I-\Lambda) U(t)\Lambda \| &\le  \|(I-\Lambda) U(t)P(0^+)\|+   \|(I-\Lambda) U(t)\left[ \Lambda-P(0^+) \right] \|\\ &\le  \|(I-\Lambda) U(t)P(0^+)\|+   \| \Lambda-P(0^+) \|\ ,
\end{align}
\ees
where to get the last inequality we used the fact that $I-\Lambda$ is also a projector and so $\|I-\Lambda\|=1$, together with the unitary invariance of the operator norm. Similarly, we find
\bes
\begin{align}
\|(I-\Lambda) U(t)P(0^+)\|  &\le  \| [I-P(t)] U(t)P(0^+)\|+\| [P(t)-\Lambda] U(t)P(0^+)\|\\ 
&\le  \| [I-P(t)] U(t)P(0^+)\| + \| P(t)-\Lambda \|   .
\end{align}
\ees
Combining these two bounds we find
\bes
\begin{align}
\|(I-\Lambda) U(t)\Lambda \| &\le    \| [I-P(t)] U(t)P(0^+)\| + \| P(t)-\Lambda \|  +   \| \Lambda-P(0^+) \|\ \label{bound43}.
\end{align}
\ees
Then, we observe that
\bes
\begin{align}
  \| [I-P(t)] U(t)P(0^+)\|&=  \| U(t)P(0^+)-P(t) U(t)P(0^+)\|= \| \left[ U(t)P(0^+)-P(t) U(t) \right] P(0^+)\|\\ &\le  \| U(t)P(0^+)-P(t) U(t)\|\ .
\end{align}
\ees
Substituting this into bound~\eqref{bound43} we find
\bes
\begin{align}
\label{bound_tri}
\|(I-\Lambda) U(t)\Lambda \| &\le    \| U(t)P(0^+)-P(t) U(t)\|  + \| P(t)-\Lambda \|  +   \| \Lambda-P(0^+) \|\  .
\end{align}
\ees
The above bound holds for any projector $\Lambda$. Choosing $\Lambda=P_0$, we can easily see that bound \eqref{bound-main-time2} follows immediately using this bound together with lemmas \ref{First_lem}, \ref{main:lemma} and \ref{lem:S}.

\subsection{Proof of the second part of Theorem~\ref{Thm:time:closed} [bound~\eqref{bound-main-time}]}
\label{app:Thm:time:closed-proof-part2}

\begin{proof}
First, without loss of generality we assume the minimum energy in $\mc{I}$ is zero, and so the maximum eigenvalue in $\mc{I}$ is $\delta E$.\footnote{If this is not the case, we can always shift the Hamiltonian by a constant and make the minimum energy  in $\mc{I}$ equal to zero.}
Let $\mc{I}(t)\subseteq \mathbb{R}$ be the interval obtained from $\mc{I}$ by adding the margin $\|V(t)\|$ on the left and the right of $\mc{I}$. Let ${P}(t)$ be the projector onto the subspace spanned by the eigenstates of $H(t)$ whose eigenvalues are in $\mc{I}(t)$. Then, since by assumption $ 2\|V(t)\|<\Delta E$, it follows from Lemma~\ref{First_lem} that the rank of $P(t)$ is the same as the rank of $P_0$, the projector onto the eigenstates of $H_0$ with eigenvalues in $\mc{I}$. Furthermore, Lemmas~\ref{lem:23} and \ref{First_lem} imply 
 \beq
 \label{bb2}
 \|Q(t) P_0\|= \|P(t) Q_0\|=  \|P(t)-P_0\|\le \frac{{2}  \|V(t)\|}{\Delta E}\ ,
 \eeq
where $Q(t)=I-P(t)$ and $Q_0=I-P_0$\ .
  
Define the Hamiltonian $H_\tru(t)$ to be the truncated version of $H(t)$ in which we have removed all the energies of $H(t)$ outside $\mc{I}(t)$, i.e.,
\beq
H_\tru(t)=P(t) H(t) P(t)\ .
\eeq
Let $U_\tru(t)$ be the unitary generated by $H_\tru(t)$, i.e.,  $dU_\tru(t)/dt=-i H_\tru(t)U_\tru(t)$, and $U_\tru(0)=I$. Let $P(0^{+})$ be the projector onto the subspace spanned by the eigenstates of $H(0^+)=H_0+V(0^{+})$ whose eigenvalues are in $\mc{I}(0^+)$, where $H_0+V(0^+)$ is the Hamiltonian of the system immediately after the perturbation is turned on. Then, we have
\bes
\begin{align}
\|{U}_\infty(t)P_0-U(t)P_0\|&\le \|U_\infty(t)P_0-U_\tru(t)P_0 \|+\|U_\tru(t)P_0-U(t)P_0 \|\\ 
&\le \|U_\infty(t)-U_\tru(t)\|+\|\left[{U}(t)-U_\tru(t)\right] \left[P(0^+)+\left(P_0-P(0^+)\right)\right]P_0\|\\ &\le \|U_\infty(t)-U_\tru(t)\|+\|\left[{U}(t)-U_\tru(t)\right]P(0^+)\| +2 \|P_0-P(0^+)\|\ \label{bound_Ham_Tri}\ .
\end{align}
\ees
Using Eq.~\eqref{bb2} we find
\begin{align}\label{bound_Ham_Tri2}
\|{U}_\infty(t)P_0-U(t)P_0\|\le \|U_\infty(t)-U_\tru(t)\|+\|{U}(t)P(0^+)-U_\tru(t)P(0^+)\|+ \frac{{4}  \|V(0^+)\|}{\Delta E}\ .
\end{align}
Lemma~\ref{lem:S} already provides us with an upper bound on $\|{U}(t)P(0^+)-U_\tru(t)P(0^+)\|$. 
Thus, to prove bound~\eqref{bound-main-time} it remains to find an upper bound on $\|U_\tru(t)-U_\infty(t)\|$. Indeed, we prove below that
\beq
\label{eq:E16}
\|U_\tru(t)-U_\infty(t)\|\le  {4} \int_{0^+}^{|t|} d\tau\   \frac{\|V(\tau)\|\left(\delta E+\|V(\tau)\|\right)}{\Delta E}\ ,
\eeq
Substituting this bound and bound~\eqref{bound_dot_P} into inequality~\eqref{bound_Ham_Tri2} we find
\begin{align}
\|{U}_\infty(t)P_0-U(t)P_0\|\le {4} \int_{0^+}^{|t|} d\tau\   \frac{\|V(\tau)\|\left(\delta E+\|V(\tau)\|\right)}{\Delta E}+ 2 \int_{0^+}^{|t|} d\tau \|\dot{P}(\tau)\|+ \frac{{4} \|V(0^+)\|}{\Delta E}\ .
\end{align}
Combining this with Lemma~\ref{main:lemma}, which puts a bound on the norm of $\|\dot{P}(\tau)\|$,  proves the second part of Theorem~\ref{Thm:time:closed}. It thus remains to prove bound~\eqref{eq:E16}, which we do next.

To prove bound~\eqref{eq:E16},  we recall that ${U}_\infty(t)$ and $U_\tru(t)$ are the unitaries generated by ${P}_0{H}(t){P}_0 $ and ${P}(t){H}(t){P}(t)$ respectively. So, we first find a bound on the difference of these two Hamiltonians. 
To do this, we note that
\begin{align}
\label{be1}
{P}(t){H}(t){P}(t) &={P}(t){H}(t)={P}(t){H}(t)P_0+{P}(t){H}(t)Q_0={P}(t){H}(t)P_0+{P}(t)H(t)P(t)Q_0\ .
\end{align}
By definition $P(t)$ is the projector onto the subspace spanned by the eigenstates of $H(t)$ whose eigenvalues are in $\mc{I}(t)$. Recall that $\mc{I}(t)$ is the interval obtained by adding the margin of  $\|V(t)\|$ to the interval $\mc{I}$. So, it follows that all the eigenvalues of $H(t)$ in $\mc{I}(t)$ are between  $-\|V(t)\|$ and $\delta E+\|V(t)\|$. Thus all the eigenvalues of $P(t)H(t)P(t)$ are likewise between $-\|V(t)\|$ and $\delta E+\|V(t)\|$, and so $\|P(t)H(t)P(t)\|\le \delta E+\|V(t)\|$. Therefore, using Eq.~\eqref{be1} we find
\bes
\label{boundP2}
\begin{align}
\|{P}(t){H}(t){P}(t)-{P}(t){H}(t)P_0\|&=\|{P}(t)H(t)P(t)Q_0\|\le \|{P}(t)H(t)P(t)\|\|P(t)Q_0\|\\
& \le  \left(\delta E+\|V(t)\|\right) \|P(t)Q_0\|\ .
\end{align}
\ees
Next, we observe that
\begin{align}
{P}_0{H}(t){P}_0 &={P}(t){H}(t)P_0+[ P_0-{P}(t) ] {H}(t)P_0 \ .
\end{align}
Therefore 
\bes
\label{boundP1}
\begin{align}
\| {P}_0{H}(t){P}_0-{P}(t){H}(t)P_0\| &=\|\left[P_0-{P}(t)\right]\left[V(t)P_0+H_0P_0\right]\| \\
&\le \|P_0-{P}(t)\|\|V(t)P_0+H_0P_0\|\\
&\le \|P_0-{P}(t)\| (\|V(t)\|+\|H_0P_0\|)  \\
\label{eq:E22d}
&\le \|P_0-{P}(t)\|(\|V(t)\|+\delta E)\\
\label{eq:E22e}
&= \|P(t)Q_0\|(\|V(t)\|+\delta E) \ ,
\end{align}
\ees
where to get inequality~\eqref{eq:E22d} we used the fact that $P_0$ is the projector onto the eigenstates of $H_0$ with energy in $\mc{I}_0$, which is between $0$ and $\delta E$,  and to get inequality~\eqref{eq:E22e} we used Lemma~\ref{lem:23}.

Combining bounds~\eqref{boundP2} and \eqref{boundP1} we find
\bes
\label{boundP8}
\begin{align}
\| {P}_0{H}(t){P}_0-{P}(t){H}(t)P(t)\| &\le\| {P}_0{H}(t){P}_0-{P}(t){H}(t)P_0\| +\|{P}(t){H}(t){P}(t)-{P}(t){H}(t)P_0\| \\
&\le 2\left(\|V(t)\|+\delta E\right)\|P(t)Q_0\| \ .
\end{align}
\ees
Using Eq.~\eqref{bb2} we have $\|P(t)Q_0\|\le\frac{{2}  \|V(t)\|}{\Delta E}$, and therefore
\begin{align}
\| {P}_0{H}(t){P}_0-{P}(t){H}(t)P(t)\|&\le {4} \frac{\|V(t)\|(\delta E+\|V(t)\|)}{\Delta E}
\end{align}
Finally, since ${P}(t){H}(t)P(t)$ generates the unitary $U_\tru(t)$ and  ${P}_0{H}(t){P}_0$ generates the unitary  ${U}_\infty(t)$ with the initial conditions $U_\infty(0)=U_\tru(0)=I$,  using Lemma \ref{lemma3} we have $\|{U}_\infty(t)P_0-U(t)P_0\|\le \int_0^{|t|}\ \| {P}_0{H}(\tau){P}_0-{P}(\tau){H}(\tau)P(\tau)\|$, and bound~\eqref{eq:E16} follows, as claimed.
\end{proof}

\section{Proof of the preliminary lemmas}
\label{app:lemma-proofs}

\subsection{Proof of Lemma~\ref{lemma3}}
\label{app:lemma3-proof}
The proof of this Lemma is given, e.g., in Ref.~\cite{Ng:2011dn} and is reproduced here for completeness.

\begin{proof}
The proof is a straightforward application of the unitary invariance of the operator norm together with the triangle inequality:
\bes
\begin{align}
\|U_1(t)-U_2(t)\| &= \|U_2^\dag(t)U_1(t)-I\|\\ &\le \|\int_{0}^{t} ds \frac{d}{ds} (U_2^\dag(s)U_1(s)) \|
\\ &\le   \|\int_{0}^{t} ds  \left[U_2^\dag(s) H_2(s) U_1(s)- U_2^\dag(s) H_1(s) U_1(s) \right] \|\\ &\le  \int_{0}^{|t|} ds  \|\left[U_2^\dag(s) H_2(s) U_1(s)- U_2^\dag(s) H_1(s) U_1(s) \right] \| \\ &\le  \int_{0}^{|t|} ds  \| H_2(s)-H_1(s)\|\ .
\end{align}
\ees
\end{proof}

\subsection{Proof of Lemma~\ref{lem:S} }
\label{app:Pdot-proof}

\subsubsection{Proof of bound~\eqref{bound_dot_Pi}} 
\begin{proof}
Let $S(t) = U^\dagger(t)P(t) U(t)$. Then
\bes
\begin{align}
\|U(t) P({0^+})- P(t)U(t)\| &= \| U^\dagger(t)P(t)U(t) - P({0^+})\| = \| S(t) - S(0^+) \| = \left\| \int_{0^+}^t \dot{S}(\tau) \ d\tau \right\| \\
& = \left\| \int_{0^+}^t \left( i U^\dagger(\tau) F(\tau) P(\tau) U(\tau) + U^\dagger(\tau) \dot{P} U(\tau) -i U^\dagger(\tau) P(\tau) H(\tau) U(\tau) \right) \ d\tau \right\| \\
\label{E11c}
& = \left\| \int_{0^+}^t U^\dagger(\tau) \dot{P}(\tau) U(\tau)  \ d\tau \right\| \\
&\leq \int_{0^+}^{|t|} \| U^\dagger(\tau) \dot{P}(\tau) U(\tau)\| \ d\tau = \int_{0^+}^{|t|} \|  \dot{P}(\tau) \| \ d\tau \,
\end{align}
\ees
where to get Eq.~\eqref{E11c} we used the premise that $[P(\tau), H(\tau)]=0$ $\forall \tau$.
\end{proof}

\subsubsection{Proof of bound~\eqref{bound_dot_P}} 
\begin{proof}
First, using the triangle inequality we find
\bes
\begin{align}
\|U(t) P({0^+})- U_\tru(t)P(0^+)\| & \leq  \|U(t) P({0^+})- P(t)U_\tru(t)\| + \|P(t)U_\tru(t) - U_\tru(t)P(0^+)\| \\
&= \|U^\dagger(t)P(t)U_\tru(t) - P({0^+})\| + \|U^\dagger_\tru(t)P(t)U_\tru(t) - P(0^+)\|\ .
\end{align}
\ees
By bound~\eqref{bound_dot_Pi} we already know that the second term satisfies
\beq
\|U^\dagger_\tru(t)P(t)U_\tru(t) - P(0^+)\| \leq \int_{0^+}^{|t|} \|\dot{P}(\tau)\| d\tau\ .
\eeq 
Now let $S(t) = U^\dagger(t)P(t)U_\tru(t)$. Then
\bes
\begin{align}
\|U^\dagger(t)P(t)U_\tru(t) - P({0^+})\| &=  \| S(t) - S(0^+) \| = \left\| \int_{0^+}^{|t|} \dot{S}(\tau) \ d\tau \right\| \\
& =  \left\| \int_{0^+}^{|t|}   \left( iU^\dagger(\tau)H(\tau) P(\tau) U_\tru(\tau) + U^\dagger(\tau) \dot{P}(\tau) U_\tru(\tau) -i U^\dagger(\tau)P(\tau)H_\tru(\tau)U_\tru(\tau)
\right)   \ d\tau \right\| \\
& =  \left\| \int_{0^+}^{|t|}   \left( iU^\dagger(\tau)H_\tru(\tau) U_\tru(\tau) + U^\dagger(\tau) \dot{P}(\tau) U_\tru(\tau) -i U^\dagger(\tau)H_\tru(t)U_\tru(\tau)
\right)   \ d\tau \right\|\label{bound645} \\
& \leq  \int_{0^+}^{|t|}   \|U^\dagger(\tau) \dot{P}(\tau) U_\tru(\tau) 
\|   \ d\tau =  \int_{0^+}^{|t|}  \| \dot{P}(\tau)\| \ d\tau \ ,
\end{align}
\ees
where to get bound \eqref{bound645} we used the definition $H_\tru(\tau)=P(t)H(t)P(t)=P(t)H(t)$. 
Combining these bounds we have the claimed result:
\begin{align}
\|U(t) P({0^+})- U_\tru(t)P(0^+)\|  \leq 2\int_{0^+}^{|t|}  \|  \dot{P}(\tau)\| \ d\tau \ .
\end{align}
\end{proof}

\subsection{Proof of Lemma~\ref{lem:23}}
\label{app:lemma23-proof}
\begin{proof}
First note that since  the supports of $P$ and $\widetilde{P}$ have the same dimension, then there exists a unitary $U$ such that $UPU^{\dag}=\widetilde{P}$, and therefore $UQU^{\dag}=\widetilde{Q}$. This implies
\beq
\|Q\widetilde{P}\|=\|QUPU^{\dag}\|= \|QUP\|.
\eeq
Using the fact that $P=PU^{\dag}(P+Q)UP$ we find that
\begin{align}
 \|P-PU^{\dag}PUP\|&=\|PU^{\dag}QUP\|= \|QUP\|^{2} =\|QUPU^{\dag}\|^{2}=\|Q\widetilde{P}\|^{2}\ ,
\end{align}
where we have used the fact that $\|AA^{\dag}\|=\|A\|^{2}$ for any operator $A$. Next, we use the fact that for any operator $A$, the operators $AA^{\dag}$ and $A^{\dag}A$ have the same eigenvalues. This implies that $PU^{\dag}PUP$ and $PUPU^{\dag}P$ have the same eigenvalues. It follows that
\beq
\|P-PUPU^{\dag}P\|= \|P-PU^{\dag}PUP\|=\|Q\widetilde{P}\|^{2}.
\eeq
Then, using the fact that $P=PU(P+Q)U^{\dag}P$ we find 
\beq
\|PUQU^{\dag}P\|=\|P-PUPU^{\dag}P\|=\|Q\widetilde{P}\|^{2},
\eeq
which implies 
\beq
\|PUQ\|=\sqrt{ \|PUQU^{\dag}P\|}=\|Q\widetilde{P}\|.
\eeq
The left-hand side is equal to $\|PUQ\|=\|PUQU^{\dag}\|=\|P\widetilde{Q}\|$. Therefore, we find $\|P\widetilde{Q}\|=\|Q\widetilde{P}\|$. 

To prove that $\|P-\widetilde{P}\|=\|P\widetilde{Q}\|=\|Q\widetilde{P}\|$, note that
\bes
\begin{align}
\|P-\widetilde{P}\|&= \|P-(P+Q)\widetilde{P}\|=\|P-P\widetilde{P}-Q\widetilde{P}\|=\|P\widetilde{Q}-Q\widetilde{P}\|=\sqrt{\| \widetilde{Q} P\widetilde{Q}+\widetilde{P} Q\widetilde{P} \|} \ ,
\end{align}
\ees
where to get the last equality we have used the fact that for any operator $A$, $\|A\|=\sqrt{\|A^\dag A \|}$. Since the supports of $\widetilde{P} Q\widetilde{P} $ and $\widetilde{Q} P\widetilde{Q}$ are orthogonal, it follows that 
\beq
\sqrt{\| \widetilde{Q} P\widetilde{Q}+\widetilde{P} Q\widetilde{P} \|}=\sqrt{\max\{\|\widetilde{Q} P\widetilde{Q}\| , \|\widetilde{P} Q\widetilde{P}\| \} }=\max\{\sqrt{\|\widetilde{Q} P\widetilde{Q}\|} , \sqrt{\|\widetilde{P} Q\widetilde{P}}\| \}=\max\{\| P\widetilde{Q}\|, \|\widetilde{P} Q\| \}\ .
\eeq
Since $\| P\widetilde{Q}\|=\|\widetilde{P} Q\|$ we find that
\begin{align}
\|P-\widetilde{P}\|&= \max\{\| P\widetilde{Q}\|, \|\widetilde{P} Q\| \}=\| P\widetilde{Q}\|=\|\widetilde{P} Q\|\ .
\end{align}

\end{proof}

\subsection{Proof of Lemma~\ref{main:lemma}}
\label{app:main:lemma-proof}

\begin{proof}
Recall that for any pair of eigenvalues  $\lambda_1$ and $\lambda_2$ of $H_0$, if $\lambda_1\in \mc{I}_0$ and  $\lambda_2\notin \mc{I}_0$, then $|\lambda_1-\lambda_2|\ge \Delta E$. Therefore, since $\|V(t)\|< \Delta E/2$, the number of orthonormal eigenstates of $H_0+V(t)$ in $\mc{I}$ is equal to the number of orthonormal eigenstates of $H_0$ in $\mc{I}_0$. 
Thus the rank of $P(t)$ is time-independent.  
Furthermore, the eigenvalues of $H_0+V(t)$ in $\mc{I}$ are separated from the rest of the spectrum of $H_0+V(t)$ by at least $\Delta E-2\|V(t)\|$. Let $H_\tot(t)\equiv H_0+V(t)$,  $\Delta H_\tot(t)\equiv H_\tot(t+\Delta t)-H_\tot(t)$\ and
\begin{align}
\Delta H^{(\text{Diag})}_\tot(t)&\equiv P(t) \Delta H_\tot(t) P(t)+ Q(t) \Delta H_\tot(t) Q(t)\\ 
\Delta H^{(\text{Off})}_\tot(t)&\equiv P(t) \Delta H_\tot(t) Q(t)+ Q(t) \Delta H_\tot(t) P(t)
\end{align}
Thus
\beq
H_\tot(t+\Delta t)= \left[ H_\tot(t) +\Delta H^{(\text{Diag})}_\tot(t)\right]+ \Delta H^{(\text{Off})}_\tot(t) \ .
\eeq
The terms inside the bracket are block-diagonal with respect to $P(t)$ and $Q(t)$. So, in the limit where $\Delta t$ is sufficiently small such that $\|\Delta H^{(\text{Diag})}_\tot(t)\|\le \Delta E-2\|V(t)\|$, we find that $P(t)$ is also the projector onto the subspace spanned by the eigenstates of $H_\tot(t) +\Delta H^{(\text{Diag})}_\tot(t)$\ . Note that the eigenvalues of $H_\tot(t) +\Delta H^{(\text{Diag})}_\tot(t)$ whose eigenvectors are inside the support of $P(t)$ and the eigenvalues whose eigenvectors are outside the support of $P(t)$ are separated from each other by at least $\Delta E- 2\|V(t)\|-\|\Delta H^{(\text{Diag})}_\tot(t)\|$. 

Next, we use Lemma~\ref{First_lem} to find the effect of adding $\Delta H^{(\text{Off})}_\tot(t)$ to $H_\tot(t) +\Delta H^{(\text{Diag})}_\tot(t)$. According to Lemma~\ref{First_lem}
\beq
\label{Eq:Deriv}
\|{P}(t+\Delta t)-P(t)\|\le \frac{{2}  \|\Delta H^{(\text{Off})}_\tot(t) \|}{\Delta E- 2\|V(t)\|-\|\Delta H^{(\text{Diag})}_\tot(t)\|}=\frac{{2}  \|P(t)\Delta H_\tot(t) Q(t)\|}{\Delta E- 2\|V(t)\|-\|\Delta H^{(\text{Diag})}_\tot(t)\|} \ .
\eeq
In the limit where $\Delta$ goes to zero, this implies
\beq
\label{Eq:Deriv2}
\|\dot{P}(t)\|\le \frac{{2}  \|P(t)\dot{H}_\tot(t) Q(t)\|}{\Delta E- 2\|V(t)\|}=\frac{{2}  \|P(t)\dot{V}(t) Q(t)\|}{\Delta E- 2\|V(t)\|}  \ .
\eeq
This proves bound~\eqref{Eq1:lema2}. Bound~\eqref{Eq2:lema2} follows since the largest eigenvalue of the off-diagonal part of $\|\dot{V}\|$ cannot be larger than the largest eigenvalue of $\|\dot{V}\|$. 

We note that, alternatively, one can prove Lemma~\ref{main:lemma} using standard resolvent techniques, i.e., using  
$P(t)=-\frac{1}{2\pi i}\int_\Gamma R(z,t) dz$,
and
$\dot{P}(t)=-\frac{1}{2\pi i}\int_\Gamma R(z,t) \dot{V}(t) R(z,t)  dz$,
where $R(z,t)\equiv (H_0+V(t)-z I)^{-1}$ is the resolvent, and $\Gamma$ is a properly chosen contour in the complex plane around $\mathcal{I}(t)$. See, e.g., \cite{Reichardt:2004}.
\end{proof}

\section{Proof of remaining statements in the main text}

\subsection{The rate of energy exchange between the system and the reservoir and the interpretation $\Omega(t)$}
\label{app:comment-energy-exchange}
In the main text we claimed that in the special case where 
\begin{enumerate}
\item $H_\B(t)=H_\B$ and hence also $\Omega(t)=\Omega$ (time-independent), and
\item $\mc{C}$ is the  bottom (top) energy sector, i.e., all energy levels in $\mc{C}^\perp$ have energy less (larger) than energy levels in $\mc{C}$,
\end{enumerate}
we can interpret $\Omega^{-1}$ as the minimum time the reservoir needs to transfer the required energy to move the system from $\mc{C}$ to  $\mc{C}^\perp$. Here we present the argument in more details. 

The key point in the argument is that in the time-independent case the total energy of system and reservoir together is a conserved quantity. So, if the energy of the system changes considerably, the energy of reservoir $\<E_B\> = \Tr[ \rho_\SB(t) H_\B]$ should also change, where $\rho_\SB(t)$ is the joint state. But the maximum rate of change of energy of reservoir is equal $\left|\frac{d}{dt}\<E_B\>\right|\le \|[H_\inter , H_\B]\|$, with equality for some $\rho_\SB(t)$. (This  follows easily from multiplying $\dot{\rho}_\SB = -i[H_\tot,\rho]$ by $H_\B$, taking the trace and using its cyclic property to cancel the terms not involving $H_\inter$, and using the inequality $|\Tr(AB)| \leq \|A\|\|B\|_1$ \cite{Bhatia:book}.)

Next we notice that the second assumption  implies that to move the system state from $\mc{C}$ to $\mc{C}^\perp$, the reservoir energy should change at least by $\Delta E-2\|H_\inter \|$ (note that adding $H_\inter $ to $H_\sys$ can shift each eigenvalue of $H_\sys$ by at most $\|H_\inter \|$, and so the gap can shrink to $\Delta E-2\|H_\inter\|$).  Since $\|[H_\inter , H_\B]\|$ is the maximum rate of  change of $\<E_B\>$, exchanging this amount of energy takes a time of at least $(\Delta E-2\|H_\inter\|)/\|[H_\inter , H_\B]\|$. So, under these assumptions $\Omega^{-1}$ (up to a factor of $1/2$) can be interpreted as the minimum time it takes the reservoir to transfer (absorb) the required energy to go from $\mc{C}$ to  $\mc{C}^\perp$. This explains the factor $t\Omega$  in bound~\eqref{bound_prob1} [and bound \eqref{main_bound_open}] in the time-independent case.

However, this simple argument for the minimum energy exchange time cannot explain some important aspects of bounds~\eqref{bound_prob1} and \eqref{main_bound_open}. In particular, it does not apply when $\mc{C}$ is neither the bottom nor the top energy sector. In this case the leakage can happen without any change in the average energy of the system and so one might expect that leakage can happen in a much shorter time. However, bound~\eqref{bound_prob1} shows that this is not the case. Nor does the simple argument explain the fact that the bound on the probability is quadratic in $\|H_\inter \|$, or the fact that the bound holds even for time-dependent Hamiltonians, where the energy is not generally conserved.

\subsection{Modification of $H_\sys$ by the reservoir as a beneficial effect}
\label{app:comment-K_sys}

In the main text we commented that in some cases the modification of $H_\sys$ by the reservoir can be a beneficial effect. This happens when $P_0 H_\inter P_0=P_{\mc{C}} K_\sys P_\mc{C}\otimes I_\B$ for some Hermitian system operator $K_\sys$. In this case, in the $\Delta E \rightarrow \infty$ limit the evolution of states inside $\mc{C}$ is described by the Hamiltonian $H_\sys+K_\sys$. This modification of the system Hamiltonian can be useful, e.g., in the context of quantum computation driven by dissipation \cite{Zanardi:2014fr}. We can use the argument leading to the bound \eqref{main_bound_open} to find how close the actual evolution is to the ideal evolution generated by $H_\sys+K_\sys$. Let $\tilde{\rho}_\id(t)=e^{-i t (H_\sys+K_\sys)} \rho(0) e^{i t (H_\sys+K_\sys)}$. Then, it turns out that for any state $\rho(0) \in \mc{C}$: 
\begin{align} 
\sqrt{ 1-F(\tilde{\rho}_\id(t),\rho(t) )}\le \frac{2}{\sqrt{2}}\left[\frac{\|H_\inter \|}{\ \Delta E}+ t \left(\frac{\|[H_\inter ,H_\B]\|}{\Delta E-2 \|H_\inter \|}+ \frac{\|H_\inter \|{(\delta E+\|H_\inter \|)}}{{\Delta E}}\right) \right]\ ,
\end{align} 
where we considered the case of time-independent $H_\B$ for simplicity. We note that the term $\frac{\|[H_\inter ,H_\B]\|}{\Delta E-2 \|H_\inter \|}$ does not appear in the Markovian analysis of Ref.~\cite{Zanardi:2014fr}.

\subsection{Proof of the inequality $F\!\left[\rho(t),\rho_\id(t)\right]\le \sqrt{1-p_\leak(t)}$}
\label{app:F-vs-p_leak}
Recall that we defined the leakage probability as $p_\leak(t)\equiv\Tr[\rho(t) Q_\mc{C}]$. Consider two states $\sigma$ [i.e., $\rho_\id(t)$] and $\tau$ [i.e., $\rho(t)$] where $P_\mc{C}\sigma P_\mc{C}=\sigma$ for some projector $P_\mc{C}$ ($P_\mc{C}=P_\mc{C}^2$, $Q_\mc{C}=I-P_\mc{C}$) with support $\mc{C}$. Considering the Taylor expansion of $\sqrt{\sigma}$, we can easily see that
\beq
\sqrt{\sigma}=\sqrt{\sigma}P_\mc{C}=P_\mc{C}\sqrt{\sigma}\ .
\eeq 
Define $\tau' \equiv P_\mc{C} \tau P_\mc{C} / \Tr(P_\mc{C} \tau)$. Then, 
\bes
\begin{align}
F(\tau,\sigma) &= \Tr(\sqrt{\sqrt{\sigma}\tau\sqrt{\sigma}}) = \Tr(\sqrt{\sqrt{\sigma}P_\mc{C}\tau P_\mc{C}\sqrt{\sigma}}) \\
&= \sqrt{\Tr(\tau P_\mc{C})}  \Tr\left(\sqrt{\sqrt{\sigma}   \frac{P_\mc{C} \tau P_\mc{C}}{\Tr(P_\mc{C} \tau)}\sqrt{\sigma}}\right) \\
& =  \sqrt{\Tr(\tau P_\mc{C})}  F(\tau',\sigma) \le \sqrt{\Tr(\tau P_\mc{C})} = \sqrt{1-\Tr(\tau Q_\mc{C})}\ .
\end{align}
\ees
In our case $\rho_\id(t) = P_\mc{C} \rho_\id(t) P_\mc{C}$ and $\Tr(\tau Q_\mc{C}) = p_\leak(t)$.

\end{document}